\documentclass[11pt]{article}
\usepackage[utf8]{inputenc}
\usepackage[T1]{fontenc}
\usepackage{lmodern}
\usepackage[DIV=11]{typearea} 
\usepackage{amssymb}
\usepackage{amsmath}
\usepackage{amsthm}
\usepackage{thmtools}
\usepackage{enumerate}

\usepackage{microtype}

\usepackage[utf8]{inputenc}

\usepackage[basic]{complexity}
\usepackage[ruled,vlined,linesnumbered,nokwfunc]{algorithm2e_}

\usepackage{tikz}
\usetikzlibrary{arrows,shapes}
\usetikzlibrary{decorations,shadows} 
\definecolor{lightgray}{rgb}{0.85,0.85,0.85}
\tikzstyle{player1}=[draw, thick, circle, fill=lightgray,inner sep=2pt, minimum width=12pt]
\tikzstyle{player2}=[draw, thick, fill=lightgray,inner sep=3pt, minimum width=15pt,minimum height=15pt]
\tikzstyle{random}=[draw, thick, diamond, rounded corners, fill=lightgray,inner sep=2pt, minimum width=12pt]
\tikzstyle{vplayer1}=[player1,inner sep=0.5pt]
\tikzstyle{vplayer2}=[player2,inner sep=0.5pt]
\tikzstyle{vrandom}=[random,inner sep=0.5pt]
\tikzstyle{p1}=[player1,minimum size=17pt]
\tikzstyle{p2}=[player2,minimum size=15pt]
\tikzstyle{arrow}=[->,line width=1pt,>=stealth',bend angle=15]

\usepackage{xcolor}
\usepackage{xspace}
\usepackage{paralist}
\usepackage{booktabs}
\usepackage{multirow}

\usepackage[
	style=alphabetic,
	backref=true,
	doi=false,
	url=false,
	maxcitenames=3,
	mincitenames=3,
	maxbibnames=10,
	minbibnames=10,
	backend=bibtex8,
	sortlocale=en_US,
	sorting=nyt,
	bibencoding=ascii
]{biblatex}

\usepackage[ocgcolorlinks]{hyperref}

\makeatletter
\g@addto@macro\bfseries{\boldmath}
\makeatother

\makeatletter
\g@addto@macro\mdseries{\unboldmath}
\g@addto@macro\normalfont{\unboldmath}
\g@addto@macro\rmfamily{\unboldmath}
\g@addto@macro\upshape{\unboldmath}
\makeatother

\DeclareCiteCommand{\citem}
    {}
    {\mkbibbrackets{\bibhyperref{\usebibmacro{postnote}}}}
    {\multicitedelim}
    {}

\renewcommand*{\multicitedelim}{\addcomma\space}

\AtEveryCitekey{\clearfield{note}}

\newcommand{\myhref}[1]{%
  \iffieldundef{doi}
    {\iffieldundef{url}
       {#1}
       {\href{\strfield{url}}{#1}}}
    {\href{http://dx.doi.org/\strfield{doi}}{#1}}%
}

\DeclareFieldFormat{title}{\myhref{\mkbibemph{#1}}}
\DeclareFieldFormat
  [article,inbook,incollection,inproceedings,patent,thesis,unpublished]
  {title}{\myhref{\mkbibquote{#1\isdot}}}

\makeatletter
\AtBeginDocument{%
    \newlength{\temp@x}%
    \newlength{\temp@y}%
    \newlength{\temp@w}%
    \newlength{\temp@h}%
    \def\my@coords#1#2#3#4{%
      \setlength{\temp@x}{#1}%
      \setlength{\temp@y}{#2}%
      \setlength{\temp@w}{#3}%
      \setlength{\temp@h}{#4}%
      \adjustlengths{}%
      \my@pdfliteral{\strip@pt\temp@x\space\strip@pt\temp@y\space\strip@pt\temp@w\space\strip@pt\temp@h\space re}}%
    \ifpdf
      \typeout{In PDF mode}%
      \def\my@pdfliteral#1{\pdfliteral page{#1}}%
      \def\adjustlengths{}%
    \fi
    \ifxetex
      \def\my@pdfliteral #1{}%
      \def\adjustlengths{\setlength{\temp@h}{-\temp@h}\addtolength{\temp@y}{1in}\addtolength{\temp@x}{-1in}}%
    \fi%
    \def\Hy@colorlink#1{%
      \begingroup
        \ifHy@ocgcolorlinks
          \def\Hy@ocgcolor{#1}%
          \my@pdfliteral{q}%
          \my@pdfliteral{7 Tr}%
        \else
          \HyColor@UseColor#1%
        \fi
    }%
    \def\Hy@endcolorlink{%
      \ifHy@ocgcolorlinks%
        \my@pdfliteral{/OC/OCPrint BDC}%
        \my@coords{0pt}{0pt}{\pdfpagewidth}{\pdfpageheight}%
        \my@pdfliteral{F}%
        \my@pdfliteral{EMC/OC/OCView BDC}%
        \begingroup%
          \expandafter\HyColor@UseColor\Hy@ocgcolor%
          \my@coords{0pt}{0pt}{\pdfpagewidth}{\pdfpageheight}%
          \my@pdfliteral{F}%
        \endgroup%
        \my@pdfliteral{EMC}%
        \my@pdfliteral{0 Tr}%
        \my@pdfliteral{Q}%
      \fi
      \endgroup
    }%
}
\makeatother

\colorlet{DarkRed}{red!50!black}
\colorlet{DarkGreen}{green!50!black}
\colorlet{DarkBlue}{blue!50!black}

\hypersetup{
	linkcolor = DarkRed,
	citecolor = DarkGreen,
	urlcolor = DarkBlue,
	bookmarksnumbered = true,
	linktocpage = true
}

\makeatletter
\newsavebox{\@brx}
\newcommand{\llangle}[1][]{\savebox{\@brx}{\(\m@th{#1\langle}\)}%
  \mathopen{\copy\@brx\mkern2mu\kern-0.9\wd\@brx\usebox{\@brx}}}
\newcommand{\rrangle}[1][]{\savebox{\@brx}{\(\m@th{#1\rangle}\)}%
  \mathclose{\copy\@brx\mkern2mu\kern-0.9\wd\@brx\usebox{\@brx}}}
\makeatother

\newcommand{\set}[1]{\{#1\}}
\newcommand{\lu}{\textup{(}}
\newcommand{\ru}{\textup{)}}
\newcommand{\upbr}[1]{\lu #1\ru}

\newcommand{\at}{\mathit{Attr}}

\newcommand{\Inf}{\mathrm{Inf}}

\newcommand{\objsty}[2]{\textrm{#1}\left(#2\right)}
\newcommand{\objstytxt}[2]{\textrm{#1}(#2)}

\newcommand{\pat}{\omega\xspace}
\newcommand{\Path}{\Omega\xspace}
\newcommand{\straa}{\sigma\xspace}
\newcommand{\Straa}{\Sigma\xspace}
\newcommand{\strab}{\pi\xspace}
\newcommand{\Strab}{\Pi\xspace}
\newcommand{\obj}{\psi\xspace}
\newcommand{\sseq}{\langle v_0,v_1,v_2,\ldots\rangle}

\newcommand{\vo}{V_1\xspace}

\newcommand{\target}{T\xspace}
\newcommand{\intarget}{\expandafter\MakeLowercase\expandafter{\target}\xspace}
\newcommand{\ec}{X\xspace}
\newcommand{\inec}{\expandafter\MakeLowercase\expandafter{\ec}\xspace}
\newcommand{\scc}{C\xspace}
\newcommand{\inscc}{\expandafter\MakeLowercase\expandafter{\scc}\xspace}

\newcommand{\Out}{\mathit{Out}}
\newcommand{\In}{\mathit{In}}
\newcommand{\OutDeg}{\mathit{Outdeg}}
\newcommand{\InDeg}{\mathit{Indeg}}

\newcommand{\pl}{{p}}
\newcommand{\op}{{\bar{p}}}

\newcommand{\vt}{V_2\xspace}
\newcommand{\ato}{\mathit{Attr}_1}
\newcommand{\att}{\mathit{Attr}_2}
\newcommand{\game}{\mathcal{G}}

\newcommand{\best}{best_\ell}
\newcommand{\cnt}{cnt_\ell}
\newcommand{\incr}{incr^\ell_v}

\newcommand{\cif}{\text{if }}

\newcommand{\domsize}{h}
\newcommand{\kmax}{\domsize_{\text{max}}}
\DeclareMathOperator{\domalg}{\ProcNameSty{kGenB{\"u}chiDominion}}
\DeclareMathOperator{\buchialg}{\ProcNameSty{GenB{\"u}chiGame}}
\DeclareMathOperator{\progress}{\ProcNameSty{GenB{\"u}chiProgressMeasure}}

\theoremstyle{plain}

\declaretheorem[numberwithin=section]{theorem}
\declaretheorem[numberlike=theorem]{lemma}
\declaretheorem[numberlike=theorem]{corollary}
\declaretheorem[numberlike=theorem]{remark}

\newtheorem{proposition}[theorem]{Proposition}
\newtheorem{reduction}[theorem]{Reduction}
\newtheorem{conjecture}[theorem]{Conjecture}
\newtheorem{invariant}[theorem]{Invariant}

\DontPrintSemicolon
\SetAlCapFnt{\normalfont\scshape}
\SetProcNameSty{texttt}
\SetFuncSty{texttt}

\usepackage{authblk}
\title{Conditionally Optimal Algorithms for Generalized Büchi Games}
\author[1]{Krishnendu Chatterjee}
\author[2]{Wolfgang Dvo{\v r}{\' a}k}
\author[2]{Monika Henzinger}
\author[2]{Veronika~Loitzenbauer}
\affil[1]{IST Austria}
\affil[2]{University of Vienna, Faculty of Computer Science, Vienna, Austria}
\date{}

\hypersetup{
	pdftitle = {Conditionally Optimal Algorithms for Generalized Büchi Games},
	pdfauthor = {Krishnendu Chatterjee, Wolfgang Dvo{\v r}{\' a}k, Monika Henzinger, Veronika Loitzenbauer}
}

\begin{document}
\maketitle
 % intro

\begin{abstract}
Games on graphs provide the appropriate framework to study several 
central problems in computer science, such as the verification and synthesis 
of reactive systems.
One of the most basic objectives for games on graphs is the liveness (or Büchi)
objective that given a target set of vertices requires that some vertex
in the target set is visited infinitely often.
We study generalized B\"uchi objectives (i.e., conjunction of liveness objectives),
and implications between two generalized B\"uchi objectives (known as GR(1) 
objectives), that arise in numerous applications in computer-aided verification.
We present improved algorithms and conditional super-linear lower bounds 
based on widely believed assumptions about the complexity of 
(A1) combinatorial Boolean matrix multiplication and (A2) CNF-SAT. 
We consider graph games with $n$ vertices, $m$ edges, and 
generalized B\"uchi objectives with $k$ conjunctions.
First, we present an algorithm with running time $O(k \cdot n^2)$, 
improving the previously known $O(k \cdot n \cdot m)$ and $O(k^2 \cdot n^2)$ 
worst-case bounds.
Our algorithm is optimal for dense graphs under (A1).
Second, we show that the basic algorithm for the problem is 
optimal for sparse graphs when the target sets have constant size under (A2). 
Finally, we consider GR(1) objectives, with $k_1$ conjunctions in the 
antecedent and $k_2$ conjunctions in the consequent, and present an 
$O(k_1 \cdot k_2 \cdot n^{2.5})$-time algorithm, improving the previously 
known $O(k_1 \cdot k_2 \cdot n \cdot m)$-time algorithm for $m > n^{1.5}$.
\end{abstract}

\section{Introduction}

\noindent{\em Games on graphs.}
Two-player games on graphs, between player~1 and the adversarial player~2, 
are central in many problems in computer science, especially in the formal analysis
of reactive systems, where vertices of the graph represent states of the 
system, edges represent transitions, infinite paths of the graph represent 
behaviors (or non-terminating executions) of the system, and the two players 
represent the system and the environment, respectively.
Games on graphs have been used in many applications related to the verification 
and synthesis of systems, such as, the synthesis of systems from specifications 
and controller-synthesis~\cite{Church62,PnueliRosner89,RamadgeWonham87},  
the verification of open systems~\cite{AHK02}, checking interface 
compatibility~\cite{InterfaceAutomata}, well-formedness of 
specifications~\cite{Dill89book}, and many others. 
We will distinguish between results most relevant for \emph{sparse graphs}, 
where the number of edges~$m$ is roughly proportional to the number of vertices~$n$,
and \emph{dense graphs} with $m = \Theta(n^2)$.
Sparse graphs arise naturally in program verification, as control-flow graphs are 
sparse~\cite{Thorup98,CIPG15}. Graphs obtained as synchronous product of several 
components (where each component makes transitions at each step)~\cite{KP09,CGIP16}
can lead to dense graphs.

\smallskip\noindent{\em Objectives.}
Objectives specify the desired set of behaviors of the system. 
The most basic objective for reactive systems is the {\em reachability} 
objective, and the next basic objective is the {\em B\"uchi \upbr{also called 
liveness or repeated 
reachability}} objective that was introduced in the seminal work of 
B\"uchi~\cite{Buchi60,Buchi62,BuchiLandweber69} for automata over infinite words.
B\"uchi objectives are specified with a target set 
$\target$ and the objective specifies the set of infinite paths in the graph 
that visit some vertex in the target set infinitely often. 
Since for reactive systems there are multiple requirements, a very central
objective to study for games on graphs is the conjunction of B\"uchi objectives,
which is known as \emph{generalized B\"uchi objective.}
Finally, currently a very popular class of objectives to specify behaviors
for reactive systems is called the \emph{GR\upbr{1} \upbr{generalized reactivity \upbr{1}}}
objectives~\cite{PPS06}.
A GR(1) objective is an implication between two generalized B\"uchi objectives{: 
the antecedent generalized B\"uchi objective is called the assumption and 
the consequent generalized B\"uchi objective is called the guarantee.
In other words, the objective requires that if the assumption generalized 
B\"uchi objective is satisfied, then the guarantee generalized B\"uchi 
objective must also be satisfied}.

We present a brief discussion about the significance of the objectives we 
consider, for a detailed discussion see~\cite{ChatterjeeH14}.
The conjunction of B\"uchi objectives is required to specify progress
conditions of mutual exclusion protocols, and deterministic B\"uchi automata
can express many important properties of linear-time temporal logic (LTL)
(the de-facto logic to specify properties of reactive systems)~\cite{KV05,KV98,AT04,KPB94}.
The analysis of reactive systems with such objectives naturally gives rise
to graph games with generalized B\"uchi objectives.
Finally, graph games with GR(1) objectives have been used in many applications,
such as the industrial example of synthesis of AMBA AHB 
protocol~\cite{BGJPPW07,GCH}
as well as in robotics applications~\cite{FKP05,CCGK15}.

\smallskip\noindent{\em Basic problem and conditional lower bounds.}
In this work we consider games on graphs with generalized B\"uchi 
and GR(1) objectives, and the basic algorithmic problem is to compute
the \emph{winning set}, i.e., the set of starting vertices where player~1 
can ensure the objective irrespective of the way player~2 plays; the way player~1 
achieves this is called her \emph{winning strategy}.
These are core algorithmic problems in verification 
and synthesis.
For the problems we consider, while polynomial-time algorithms are known,
there are no super-linear lower bounds.
Since for polynomial-time algorithms unconditional super-linear lower bounds 
are extremely rare in the whole of computer science, we consider 
{\em conditional lower bounds}, which assume that for some well-studied 
problem the known algorithms are optimal up to some lower-order factors. 
In this work we consider two such well-studied assumptions:
(A1)~there is no combinatorial\footnote{Combinatorial here means avoiding fast matrix multiplication~\cite{LeGall14}, 
see also the discussion in~\cite{HenzingerKNS15}.} algorithm with running time of
$O(n^{3-\varepsilon})$ for any $\varepsilon > 0$ to multiply two $n \times n$ Boolean matrices; 
or (A2)~for all $\varepsilon>0$ there exists a $k$ such that there is no algorithm 
for the $k$-CNF-SAT problem that runs in $O(2^{(1-\varepsilon) \cdot n} \cdot \operatorname{poly}(m))$ 
time, where $n$ is the number of variables and $m$ the number of clauses. 
These two assumptions have been used to establish lower bounds for 
several well-studied problems, such as dynamic graph algorithms~\cite{AbboudW14,AbboudWY15}, 
measuring the similarity of strings~\cite{AbboudWW14,Bringmann14,
BringmannK15,BackursI15,AbboudBW15b}, context-free grammar
parsing~\cite{Lee02,AbboudBW15a}, and verifying first-order graph 
properties~\cite{PatrascuW10,Williams14}. 

\smallskip\noindent{\em Our results.} 
We consider games on graphs with $n$ vertices, $m$ edges, generalized B\"uchi objectives
with $k$ conjunctions, and target sets of size $b_1, b_2, \ldots, b_k$, 
and GR(1) objectives with $k_1$ conjunctions in the assumptions and 
$k_2$ conjunctions in the guarantee.
Our results are as follows. 
\begin{itemize}
\item {\em Generalized B\"uchi objectives.}
The classical algorithm for generalized B\"uchi objectives 
requires $O(k \cdot \min_{1 \leq i \leq k} b_i \cdot m)$ time. 
Furthermore,
there exists an $O(k^2 \cdot n^2)$-time algorithm via a reduction to Büchi 
games~\cite{BloemCGHJ10,ChatterjeeH14}.
\begin{enumerate}
\item {\em Dense graphs.} 
Since $\min_{1\leq i \leq k} b_i=O(n)$ and $m=O(n^2)$,
the classical algorithm has a worst-case running time of $O(k\cdot n^3)$.
First, we present an algorithm with worst-case running time 
$O(k \cdot n^2)$, which is an improvement for instances with 
$\min_{1 \leq i \leq k} b_i \cdot m = \omega(n^2)$.
Second, for dense graphs with $m = \Theta(n^2)$ and $k = \Theta(n^c)$ for any $0<c\leq 1$
our algorithm is optimal under (A1); i.e.,
improving our algorithm for dense graphs would imply a faster 
(sub-cubic) combinatorial Boolean matrix multiplication algorithm.

\item {\em Sparse graphs.} 
We show that for $k = \Theta(n^c)$ for any $0<c\leq 1$, 
for target sets of constant size, and sparse 
graphs with $m = \Theta(n^{1+o(1)})$ 
the basic algorithm is optimal under (A2).
In fact, our conditional lower bound under (A2) holds even when each target 
set is a singleton.
Quite strikingly, our result implies that improving the basic algorithm for 
sparse graphs even with singleton sets would require a major breakthrough 
in overcoming the exponential barrier for SAT.
\end{enumerate}
In summary, for games on graphs, we present an improved algorithm for 
generalized B\"uchi objectives for dense graphs that is optimal under (A1); 
and show that under (A2) the basic algorithm is optimal 
for sparse graphs and constant size target sets.

The conditional lower bound for dense graphs means in particular that 
for unrestricted inputs the dependence of the runtime on~$n$ cannot be improved,
whereas the bound for sparse graphs makes the same statement for the dependence
on~$m$.
Moreover, as the graphs in the reductions for our lower bounds can be made acyclic by deleting a single vertex,
our lower bounds also apply to a broad range of digraph parameters. 
For instance, let $w$ be the DAG-width~\cite{BerwangerDHK06} of a graph, then there is no $O(f(w) \cdot n^{3-o(1)})$-time algorithm under (A1) and no $O(f(w)\cdot m^{2-o(1)})$-time algorithm under~(A2).

\item {\em GR(1) objectives.} 
We present an algorithm for games on graphs with GR(1) objectives that has 
$O( k_1 \cdot k_2 \cdot n^{2.5})$ running time and improves the previously 
known $O(k_1 \cdot k_2 \cdot n \cdot m)$-time algorithm~\cite{JuvekarP06} 
for $m > n^{1.5}$.
Note that since generalized B\"uchi objectives are special cases of GR(1) 
objectives, our conditional lower bounds for  generalized B\"uchi objectives 
apply to GR(1) objectives as well, but are not tight.
\end{itemize}
All our algorithms can easily be modified to also return the corresponding 
winning strategies for both players within the same time bounds.

\smallskip\noindent{\em Implications.}  
We discuss the implications of our results.
\begin{enumerate}
\item {\em Comparison with related models.}
We compare our results for game graphs to the special case of standard graphs
(i.e., games on graphs with only player~1) and the related model of 
Markov decision processes (MDPs) (with only player~1 and stochastic transitions).
First note that for reachability objectives, linear-time algorithms exist for 
game graphs~\cite{Beeri80,Immerman81}, whereas for MDPs\footnote{For MDPs the 
winning set refers to the almost-sure winning set 
that requires that the objective is satisfied with probability~1.}
the best-known algorithm has 
running time $O(\min(n^2, m^{1.5}))$~\cite{ChatterjeeJH03,ChatterjeeH14}.
For MDPs with reachability objectives, a linear or even $O(m \log n)$ time algorithm is a 
major open problem, i.e., there exist problems that seem harder for MDPs than for game 
graphs.
Our conditional lower bound results show that under assumptions 
(A1) and (A2) the algorithmic problem for generalized B\"uchi objectives 
is strictly harder for games on graphs as compared to standard graphs and MDPs.
More concretely, for $k=\Theta(n)$, 
(a)~for dense graphs ($m=\Theta(n^2)$) 
and $\min_{1\leq i \leq k} b_i=\Omega(\log n)$, our lower bound for games on graphs under (A2) 
is $\Omega(n^{3-o(1)})$, whereas both the graph and the 
MDP problems can be solved in $O(n^2)$ 
time~\cite{ChatterjeeH12,ChatterjeeH14}; and 
(b)~for sparse graphs ($m=\Theta(n^{1+o(1)})$) with 
$\min_{1\leq i \leq k} b_i=O(1)$,  
our lower bound for games on graphs under (A1) is $\Omega(m^{2-o(1)})$, 
whereas the graph problem can be solved in $O(m)$ time and the MDP problem 
in $O(m^{1.5})$ time~\cite{AlurHenzingerBook,ChatterjeeH11}; respectively.

\item {\em Relation to SAT.} We present an algorithm for game graphs
with generalized B\"uchi objectives  and show that improving the algorithm
would imply a better algorithm for SAT, and thereby establish an interesting
algorithmic connection for classical objectives in game graphs and the SAT problem.
\end{enumerate}

{
\smallskip\noindent{\em Outline.}
In Section~\ref{sec:prelim} we provide formal definitions and state the 
conjectures on which the conditional lower bounds are based. 
In Section~\ref{sec:GenBuchiAlg} we consider algorithms for generalized Büchi objectives
and first present a basic algorithm which is in $O(knm)$
time and then improve it to an $O(k \cdot n^2)$-time algorithm.
In Section~\ref{sec:lowerbounds} we provide conditional lower bounds for generalized Büchi objectives.
Finally, in Section~\ref{sec:GR1Alg} we study algorithms for games with GR(1) objective
and first give a basic algorithm which is in $O(k_1 \cdot k_2 \cdot n^3)$ time and 
then improve it to an $O(k_1 \cdot k_2 \cdot n^{2.5})$-time algorithm.}

 % prelim
\section{Preliminaries}\label{sec:prelim}
\subsection{Basic definitions for Games on Graphs}
\noindent{\em Game graphs.}
A \emph{game graph} $\game = ((V, E), (V_1, V_2))$ is a directed graph $G = (V, E)$
with a set of vertices $V$ and a set of edges $E$ and a partition of $V$
into \emph{player~1 vertices}~$V_1$ and \emph{player~2 vertices}~$V_2$.
Let $n = \lvert V \rvert$ and $m = \lvert E \rvert$.
Given such a game graph $\game$, we denote with $\overline{\game}$ the game graph where the player~1 and player~2 vertices of $\game$ are  interchanged, i.e, $\overline{\game} = ((V, E), (V_2, V_1))$. We use~$\pl$ to denote a player and $\op$ to denote
its opponent.
For a vertex $u\in V$, we write $\Out(u)=\set{v\in V \mid (u,v) \in E}$ 
for the set of successor vertices of~$u$ and $\In(u)=\set{v \in V \mid (v,u) 
\in E}$ for the set of predecessor vertices of $u$. If necessary, we refer to 
the successor vertices in a specific graph by using, e.g., $\Out(G, u)$.
We denote by $\OutDeg(u)=|\Out(u)|$ the number of outgoing edges from $u$, 
and by $\InDeg(u)=|\In(u)|$ the number of incoming edges. We assume for 
technical convenience $\OutDeg(u) \ge 1$ for all $u \in V$.

\medskip\noindent{\em Plays and strategies.}
A \emph{play} on a game graph is an infinite sequence $\pat = \sseq$ of vertices 
such that $(v_\ell, v_{\ell+1}) \in E$ for all $\ell \ge 0$. The set of all plays
is denoted by $\Path$. Given a finite prefix $\pat \in V^* \cdot V_\pl$ of a 
play that ends at a player~$\pl$ vertex~$v$, a \emph{strategy}
$\straa: V^* \cdot V_\pl \rightarrow V$ of player~$\pl$
is a function that chooses a successor vertex $\straa(\pat)$ among the vertices
of $\Out(v)$. 
We denote by $\Sigma$ and $\Pi$ the set of all strategies of player~1 and player~2 respectively.
The play $\pat(v, \straa, \strab)$ is uniquely defined by a start 
vertex~$v$, a player~1 strategy~$\straa \in \Straa$, and a player~2 strategy
$\strab \in \Strab$
as follows: 
$v_0=v$ and for all $j \geq 0$, if $v_j \in V_1$, then $v_{j+1}=\straa(\langle 
v_1,\dots,v_j\rangle)$, 
and if $v_j \in V_2$, then $v_{j+1}=\strab(\langle v_1,\dots,v_j \rangle)$.

\medskip\noindent{\em Objectives.}
An objective $\obj$ is a set of plays that is winning for a player. 
We consider zero-sum games where for a player-1 objective~$\obj$ the complementary
objective $\Omega \setminus \obj$ is winning for player~2.
In this work we consider 
only \emph{prefix independent objectives}, for which the set of desired plays 
is determined by the set of vertices~$\Inf(\pat)$ that occur
\emph{infinitely often}  in a play~$\pat$. Given a target set $\target \subseteq V$,
a play~$\pat$ belongs to the \emph{Büchi objective} $\objsty{Büchi}{\target}$
iff $\Inf(\pat) \cap \target \ne \emptyset$. For the complementary \emph{co-Büchi
objective}
we have $\pat \in \objsty{coBüchi}{\target}$ iff $\Inf(\pat) \cap \target = \emptyset$.
A \emph{generalized \upbr{or conjunctive} Büchi objective} is specified by a set of $k$ target sets
$\target_\ell$ for $1 \le \ell \le k$ and is satisfied for a play~$\pat$ iff 
$\Inf(\pat) \cap \target_\ell \ne \emptyset$ for \emph{all} $1 \le \ell \le k$. Its 
complementary objective is the \emph{disjunctive co-Büchi objective} that 
is satisfied iff $\Inf(\pat) \cap \target_\ell = \emptyset$ for \emph{one of} the 
$k$ target sets.
A \emph{generalized reactivity-1} (GR(1)) objective is specified by 
two generalized Büchi objectives, $\bigwedge_{t=1}^{k_1} \objsty{Büchi}{L_t}$
and $\bigwedge_{\ell=1}^{k_2} \objsty{Büchi}{U_\ell}$, and is satisfied if 
whenever the first generalized Büchi objective holds, then also the second 
generalized Büchi objective holds; in other words, either 
$\bigvee_{t=1}^{k_1} \objsty{coBüchi}{L_t}$ holds, or 
$\bigwedge_{\ell=1}^{k_2} \objsty{Büchi}{U_\ell}$ holds.

In this paper we specify a game by a game graph $\game$ and an objective $\obj$ 
for player~1. Player~2 has the complementary objective $\Omega \setminus \obj$.

\smallskip\noindent\emph{Winning strategies and sets.}
A strategy~$\straa$ is winning for player~$\pl$ at a start vertex~$v$ if the 
resulting play is winning for player~$\pl$ irrespective of the strategy of 
his opponent, player~$\op$, i.e., $\omega(v,\sigma,\pi) \in \obj$ for all $\pi$. 
A vertex~$v$ belongs to 
the \emph{winning set} $W_\pl$ of player~$\pl$ if player~$\pl$ has a winning 
strategy from $v$. Every vertex is winning for exactly 
one of the two players~\cite{Mar75}{ (cf.\ Theorem~\ref{th:det})}. 
When an explicit reference to a specific game $(\game, \obj)$ is required,
we use $W_\pl(\game, \obj)$ to refer to the winning sets.
{
\begin{theorem}[\cite{Mar75}]\label{th:det}
	In graph games with prefix independent objectives
	the winning sets of the two players partition the vertex set~$V$. 
\end{theorem}
} 
{For the analysis of our algorithms we further introduce the notions of 
 \emph{closed sets}, \emph{attractors}, and \emph{dominions}.}
 
\smallskip\noindent\emph{Closed sets.}
A set $U \subseteq V$ is \emph{$\pl$-closed} (in $\game$) if for all $\pl$-vertices 
$u$ in $U$ we have $\Out(u) \subseteq U$ and for all $\op$-vertices 
$v$ in $U$ there exists a vertex $w \in \Out(v) \cap U$.
Note that player~$\op$ can ensure that a play that currently ends in a $\pl$-closed 
set never leaves the $\pl$-closed set against any strategy of player~$\pl$ by 
choosing an edge $(v,w)$ with $w \in \Out(v) \cap U$ whenever the current 
vertex~$v$ is in $U \cap V_\op$~\cite{Zielonka98}.
Given a game graph $\game$ and a $\pl$-closed set $U$, we denote by $\game[U]$ 
the game graph induced by the set of vertices~$U$. Note that given that in $\game$
each vertex has at least one outgoing edge, the same property holds for $\game[U]$.
We further use the shortcut $\game \setminus X$ to denote $\game[V \setminus X]$.

\smallskip\noindent\emph{Attractors.}
In a game graph $\game$, a $\pl$-\emph{attractor} $\at_\pl(\game, U)$ of a set 
$U \subseteq V$ is the set of vertices from which player~$\pl$ has a strategy 
to reach $U$ against all strategies of player~$\op$~\cite{Zielonka98}. We have 
that $U \subseteq \at_\pl(\game, U)$. 
A $\pl$-attractor can be constructed inductively as follows: Let $R_0=U$; and 
for all $j \geq 0$ let 
{
\begin{equation*}\label{eq:attr}
	R_{j+1} = R_j \cup \set{v \in V_\pl \mid \Out(v) \cap R_j \neq \emptyset}
		      \cup \set{v \in V_\op \mid \Out(v) \subseteq R_j}.
\end{equation*}}
Then $\at_\pl(\game, U)\!=\!\bigcup_{j\ge 0}\!R_j$.

{The \emph{$\pl$-rank} of a vertex $v$ w.r.t.\ a set $U$ is given by $rank_\pl(\game,U,v)\!=\min\{j\mid v \in R_j\}$ if $v \in \at_\pl(\game, U)$
and is $\infty$ otherwise.}

\smallskip\noindent\emph{Dominions.}
A set of vertices $D \ne \emptyset$ is a player-$\pl$ \emph{dominion} if 
player~$\pl$ has a winning strategy from every vertex in $D$ that
also ensures only vertices in $D$ are visited. 
The notion of dominions was introduced by~\cite{JurdzinskiPZ08}. Note that 
a player-$\pl$ dominion is also a $\op$-closed set and the 
$\pl$-attractor of a player-$\pl$ dominion 
is again a player-$\pl$ dominion.

\smallskip
{The lemma below summarizes some well-known facts about closed sets, attractors, and winning sets.}
{\begin{lemma}\label{lem:attr}
The following assertions hold for game graphs~$\game$ where each vertex has 
at least one outgoing edge. The assertions referring to winning sets hold for 
graph games with prefix independent objectives. Let $X \subseteq V$.
\begin{compactenum}
{\item From each vertex of $\at_\pl(\game, X)$ player~$\pl$ has a memoryless 
strategy that stays within $\at_\pl(\game, X)$ to reach~$X$ against any 
strategy of player~$\op$~\cite{Zielonka98}.\label{sublem:attrstr}
\item The attractor $\at_\pl(\game, X)$ can be computed in 
$O(\sum_{v \in \at_\pl(\game, X)} \lvert\In(v)\rvert)$ 
time~\cite{Beeri80,Immerman81}.\label{sublem:attrtime}
}
\item The set $V \setminus \at_\pl(\game, X)$ is $\pl$-closed on 
$\game$~\cite[Lemma~4]{Zielonka98}.\label{sublem:complattr}
\item Let $X$ be $\pl$-closed on $\game$. Then $W_\op(\game[X]) \subseteq W_\op(\game)$~\cite[Lemma~4.4]{JurdzinskiPZ08}.\label{sublem:winclosed}
\item Let $X$ be a subset of the winning set $W_\pl(\game)$ of player~$\pl$ and 
let $A$ be its $\pl$-attractor $\at_\pl(\game, X)$. Then the winning set $W_\pl(\game)$
of the player~$\pl$ is the union of $A$ and the winning set 
$W_\pl(\game[V \setminus A])$, 
and the winning set $W_\op(\game)$ of the opponent~$\op$ is equal to 
$W_\op(\game[V \setminus A])$
~\cite[Lemma~4.5]{JurdzinskiPZ08}.
\label{sublem:subgraph}
\end{compactenum}
\end{lemma}}

\subsection{Conjectured Lower Bounds}\label{sec:conjectures}

While classical complexity results are based on assumptions about relationships between complexity classes,
e.g., $\P \ne \NP$, polynomial lower bounds are often based on widely believed,
conjectured lower bounds 
for well studied algorithmic problems.
We next discuss the popular conjectures that will be the basis for our lower bounds{ for generalized Büchi games}.

First, we consider conjectures on Boolean matrix multiplication~\cite{WilliamsW10,AbboudW14} and 
triangle detection~\cite{AbboudW14} in graphs, which build the basis for our lower bounds on dense graphs. 
A triangle in a graph is a triple $x, y, z$ of distinct vertices 
such that $(x, y), (y, z), (z, x) \in E$.

\begin{conjecture}[Combinatorial Boolean Matrix Multiplication Conjecture (BMM)]\label{conj:bmm}
There is no $O(n^{3 - \varepsilon})$ time combinatorial algorithm for computing the Boolean product of 
two $n \times n$ matrices for any $\varepsilon > 0$.
\end{conjecture}

\begin{conjecture}[Strong Triangle Conjecture (STC)]\label{conj:triangle}
There is {no $O(\min\{n^{\omega - \varepsilon},$ $m^{2 \omega/(\omega + 1) - \varepsilon}\})$ expected time algorithm and }no
$O(n^{3 - \varepsilon})$ time combinatorial algorithm that can detect whether a graph contains a triangle for any $\varepsilon > 0${, 
where $\omega < 2.373$ is the matrix multiplication exponent}. 
\end{conjecture}
{By a result of 
Vassilevska~Williams and Williams~\cite{WilliamsW10}, we have that 
BMM is equivalent to the combinatorial part of STC. Moreover, if we do not restrict 
ourselves to combinatorial algorithms, STC still gives a super-linear lower bound.}

Second, we consider the Strong Exponential Time Hypothesis~\cite{ImpagliazzoPZ01,CalabroIP09}
and the Orthogonal Vectors Conjecture~\cite{AbboudWW16}, the former dealing with satisfiability in propositional logic and the latter with 
the \emph{Orthogonal Vectors Problem}.

\emph{The Orthogonal Vectors Problem} (OV). Given two sets $S_1, S_2$ of $d$-bit 
vectors with $|S_1 |, |S_2| \leq N$ and $d \in \Theta(\log N)$, are there $u \in S_1$ 
and $v \in S_2$ such that $\sum_{i=1}^{d} u_i \cdot  v_i = 0$?

\begin{conjecture}[Strong Exponential Time Hypothesis (SETH)]\label{conj:seth}
For each  $\varepsilon >0$ there is a $k$ such that k-CNF-SAT on $n$ variables and $m$ clauses
cannot be solved in time~$O(2^{(1-\varepsilon)n} \operatorname{poly}(m))$.
\end{conjecture}

\begin{conjecture}[Orthogonal Vectors Conjecture (OVC)]\label{conj:ov}
There is no $O(N^{2-\varepsilon})$ time algorithm for the Orthogonal Vectors Problem for any $\varepsilon > 0$.
\end{conjecture}

{By a result of Williams~\cite{Williams05} 
we know that SETH implies OVC}, i.e.,
whenever a problem is hard assuming OVC, it is also hard when assuming SETH.
Hence, it is preferable to use OVC for proving lower bounds.
Finally, to the best of our knowledge, no such relations between the former two conjectures and the latter two conjectures 
are known.
\begin{remark}
	The conjectures that no \emph{polynomial} improvements over the best known
	running times are possible do not exclude improvements by sub-polynomial 
	factors such as poly-logarithmic factors or factors of, e.g., $2^{\sqrt{\log n}}$
	as in \cite{Williams14a}.
\end{remark}

 % ConjBuchiAlg
\section{Algorithms for Generalized Büchi Games}\label{sec:GenBuchiAlg}
For generalized Büchi games we first present the basic algorithm 
that follows from the results of~\cite{EJ91,McNaughton93,Zielonka98}.
The basic algorithm{ (cf.\ Algorithm \ref{alg:ConjBuchiBasic})} runs in
time $O(knm)$. We then improve it to an $O(k \cdot n^2)$-time algorithm by 
exploiting ideas from the $O(n^2)$-time algorithm for Büchi games 
in~\cite{ChatterjeeH12}. 
The basic algorithm is fast for instances where one Büchi set is  small, i.e., the algorithm runs in time 
$O(k \cdot min_{1 \leq \ell \leq k} b_\ell \cdot m)$~time, where $b_\ell=|\target_\ell|$.

{
\smallskip\noindent\emph{Reduction to Büchi Games.} 
Another way to implement generalized Büchi games is by a reduction to Büchi games as follows (see also \cite{BloemCGHJ10}).
Make $k$ copies $V^\ell$, $1 \le \ell \le k$, 
of the vertices of the original game graph and draw 
an edge $(v^j,u^j)$ if $(v,u)$ is an edge in the original graph and $v \not \in \target_j$, and
an edge $(v^j,u^{j\oplus 1})$ if $(v,u)$ is an edge in the original graph and $v \in \target_j$
(where $j\oplus 1= j+1$ for $j<k$ and $k\oplus 1 = 1$). Finally, pick the Büchi set $\target_\ell$ of minimal size and
make its copy $\target_\ell^\ell$ in $V^\ell$ the target set for the Büchi game.
This reduction results in another $O(k \cdot min_{1 \leq \ell \leq k} b_\ell \cdot m)$~time algorithm 
when combined with the basic algorithm for Büchi and in
an $O(k^2 n^2)$~time algorithm when combined with the $O(n^2)$ time algorithm for Büchi~\cite{ChatterjeeH12}.}

{\smallskip\noindent\emph{Notation.}} Our algorithms iteratively identify sets of vertices that are winning for 
player~2, i.e., player-2 dominions, and remove them from the graph.
{In the algorithms and their analysis we denote the sets in the} $j$th-iteration 
with superscript~$j$, in particular
$\game^1 = \game$, where $\game$ is the input game graph,
$G^j$ is the graph of $\game^j$, $V^j$ is the 
vertex set of $G^j$, and $\target^j_\ell = V^j \cap T_\ell$.
We also use $\set{\target^j_\ell}$ to denote the list of Büchi sets $(\target^j_1, \target^j_2, \dots, \target^j_k)$,
in particular when updating all the sets in a uniform way.

{\subsection{Basic Algorithm}}
For each set~$U$ that is closed for player 1 we have that from each vertex $u\in U$
player~2 has a strategy to ensure that the play never 
leaves $U$~\cite{Zielonka98}. Thus, if there is a Büchi set $\target_\ell$ with 
$\target_\ell \cap U = \emptyset$, then the set~$U$ is a player-2 dominion.
Moreover, if $U$ is a player-2 dominion, also the attractor $\att(\game, U)$ of $U$ 
is a player-2 dominion.
The basic algorithm{ (cf.\ Algorithm \ref{alg:ConjBuchiBasic})} proceeds as follows. 
It iteratively computes vertex sets~$S^j$ closed for player~1 that do not 
intersect with one of the Büchi sets. If such a player-2 dominion~$S^j$ is found,
then all vertices of $\att(\game^j, S^j)$ are marked as winning for player~2 and
removed from the game graph; the remaining game graph is denoted by $\game^{j+1}$. 
To find a player-2 dominion~$S^j$, for each $1 \leq \ell \leq k$ the 
attractor $Y^j_\ell = \ato(\game^j, \target^j_\ell)$ of the Büchi set
$\target^j_\ell$ is determined.
If for some~$\ell$ the complement of $Y^j_\ell$
is not empty, then we assign $S^j = V^j \setminus Y^j_\ell$ for the smallest
such~$\ell$.
The algorithm terminates if in some iteration $j$ for each $1 \leq \ell \leq k$ the attractor $Y^j_\ell$ 
contains all vertices of~$V^j$.
In this case the set~$V^j$ is returned as the winning set 
of player~1. The winning strategy of player~1 from these vertices is then
a combination of the attractor strategies to the sets $\target^j_\ell$.

{
\begin{algorithm}
	\SetAlgoRefName{GenBuchiGameBasic}
	\caption{Generalized Büchi Games in $O(k \cdot b_1 \cdot m)$ Time}
	\label{alg:ConjBuchiBasic}
	\SetKwInOut{Input}{Input}
	\SetKwInOut{Output}{Output}
	\SetKw{break}{break}
	\BlankLine
	\Input{Game graph $\game=((V, E),(\vo,\vt))$ and objective $\bigwedge_{1 \le \ell \le k} \objsty{Büchi}{\target_\ell}$}
	\Output
	{
	  Winning set of player~1 
	}
	\BlankLine
	$\game^1 \gets \game$\;
	$\set{\target^1_\ell} \gets \set{\target_\ell}$\; 
	$j \gets 0$\;
	\Repeat
	  {
	    $D^j = \emptyset$
	  }
	  {
	  $j \gets j+1$\;
	    \For{$1\leq \ell \leq k$}
	      {
					$Y^j_\ell \gets \ato(\game^j, \target^j_\ell)$\;
					$S^j \gets V^j \setminus Y^j_\ell$\;
					\lIf{$S^j \ne \emptyset$}{
						\break
					}
	      }
	    $D^j \gets \att(\game^j, S^j)$\;
	    $\game^{j+1} \gets \game^j \setminus D^j$\;
	    $\set{\target^{j+1}_\ell}\gets \set{\target^j_\ell \setminus D^j}$\; 
	  }
	  \Return $V^j$\;
\end{algorithm}

\begin{proposition}[Runtime Algorithm~\ref{alg:ConjBuchiBasic}]
	 The basic algorithm for generalized Büchi games terminates in $O(k \cdot b_1 \cdot m)$~time, where
	 $b_1=|T_1|$, and thus also in $O(knm)$~time .
\end{proposition}
\begin{proof}
	In each iteration of the repeat-until loop at most $k+1$ attractor computations
	are performed, which can each be done in $O(m)$ time. 
	We next argue that the repeat-until loop terminates after at most $2 b_1 +2$ iterations.
	We use that \upbr{a} a player-2 edge from $Y^j_\ell = \ato(\game^j, \target^j_\ell)$
	to $V^j \setminus Y^j_\ell$ has to originate from a vertex of $\target^j_\ell$
	and \upbr{b} if a player-1 attractor contains a vertex, then it contains also 
	the player-1 attractor of this vertex.
	In each iteration we have one of the following situations:
	\begin{enumerate}
	 \item $S^j = \emptyset$: The algorithm terminates.
	 \item $\ato(\game^j, \target^j_1) = V^j$ and $\ato(\game^j, \target^j_\ell) \not= V^j$ for some $\ell > 1$:	       
	       We have that $\target^j_1 \not\subseteq \ato(\game^j, \target^j_\ell)$
	       as $\target^j_1 \subseteq \ato(\game^j, \target^j_\ell)$ would imply that
	       also $\ato(\game^j, \target^j_1) = V^j \subseteq \ato(\game^j, \target^j_\ell)$
	       which is in contradiction to the assumption.
	       Thus we obtain $|\target^{j+1}_1|<|\target^j_1|$.
	 \item $\ato(\game^j, \target^j_1) \not= V^j$ and $D^j \cap \target^j_1\not=\emptyset$:
	       We immediately get $|\target^{j+1}_1|<|\target^j_1|$.
	 \item $\ato(\game^j, \target^j_1) \not= V^j$ and $D^j \cap \target^j_1=\emptyset$:
               In this case $D^j = \att(\game^j, S^j)=S^j$ and thus in the next iteration we have either situation (1) or (2).
               Notice that, for each vertex $v \in \ato(\game^j, \target^j_1)$ player 1 
               has a strategy to reach $\target^j_1$ and thus for $v$ to be in $D^j$ the set $D^j$
               has to contain at least one vertex of $\target^j_1$.
               
	\end{enumerate}
	By the above we have that $\target^{j}_1$ is decreased in at least every second iteration of the loop.
	For $\target^{j}_1 =\emptyset$ we have $\ato(\game^j, \target^j_1)= \emptyset$ and thus $V^{j+1}=\emptyset$,
	which terminates the algorithm in the next iteration.
	Thus we have 
	that each iteration takes time $O(km)$ and there are $2b_1+2$, i.e., $O(b_1)$, iterations.
\end{proof}
As we can always rearrange the Büchi sets such that $b_1=\min_{1 \leq \ell \leq k} b_\ell$ 
this gives an $O(k \cdot \min_{1 \leq \ell \leq k} b_\ell \cdot m)$ algorithm for generalized Büchi games.

For the final game graph $\game^j$ we have that all vertices are in all the player-1 attractors of Büchi sets $\target_\ell$.
Thus player~1 can win the game by following one attractor strategy until the corresponding Büchi set is reached and 
then switching to the attractor strategy of the next Büchi set.

\begin{proposition}[Soundness Algorithm~\ref{alg:ConjBuchiBasic}]\label{prop:soundness_ConjBuchiBasic}
	Let $V^{j^*}$ be the set of vertices returned by the algorithm.
Each vertex in $V^{j^*}$ is winning for player~1.
\end{proposition}
\begin{proof}
	Let the $j^*$-th iteration be the last iteration of the algorithm. 
	We have $S^{j^*} = \emptyset$. Thus each vertex of $V^{j^*}$
	is contained in $\ato(\game^{j^*}, \target^{j^*}_\ell)$ for each $1 \le \ell \le k$.
	Additionally, either $V^{j^*} = \emptyset$ or 
	$\target^{j^*}_\ell \ne \emptyset$ for all $1 \le \ell \le k$.
	Further we have that $V^{j^*}$ is closed
	for player~2 as only player~2 attractors were removed from $V$ to obtain
	$V^{j^*}$ (i.e., we apply Lemma~\ref{lem:attr}(\ref{sublem:complattr}) 
	inductively). Hence player~1 has the following winning strategy (with memory) 
	on the vertices of $V^{j^*}$: On the vertices of $V^{j^*} \setminus 
	\bigcap_{\ell = 1}^k \target^{j^*}_\ell$ first follow the attractor strategy 
	{(Lemma~\ref{lem:attr}(\ref{sublem:attrstr}))}
	for $\target^{j^*}_1$ until a vertex of $\target^{j^*}_1$ is reached, then
	the attractor strategy for $\target^{j^*}_2$ until a vertex of 
	$\target^{j^*}_2$ is reached and so on until the set $\target^{j^*}_k$ 
	is reached; then restart with $\target^{j^*}_1$. 
	On the vertices of $\bigcap_{\ell = 1}^k \target^{j^*}_\ell 
	\cap V_1$ player~1 can pick any outgoing edge whose endpoint is in $V^{j^*}$.
	Since $V^{j^*}$ is closed for player~2 and $\target^{j^*}_\ell \ne \emptyset$ 
	for all $1 \le \ell \le k$, this strategy exists,
	never leaves the set $V^{j^*}$, and satisfies the generalized Büchi objective.
\end{proof}

For completeness we use that each 1-closed set that avoids one Büchi set is winning for player 2 and that, 
by Lemma~\ref{lem:attr}(\ref{sublem:subgraph}), we can remove such sets from the game graph.

\begin{proposition}[Completeness Algorithm~\ref{alg:ConjBuchiBasic}]
	Let $V^{j^*}$ be the set returned by the algorithm.
	Player~2 has a winning strategy from each vertex in $V \setminus V^{j^*}$.
\end{proposition}
\begin{proof}
	By Lemma~\ref{lem:attr}(\ref{sublem:subgraph}) it is sufficient to show
	that, in each iteration $j$, player~2 has a winning strategy in $\game^j$ from each vertex of $S^j$.
	Let $\ell$ be such that
	$S^j = V^j \setminus \ato(\game^j, \target^j_\ell)$. By 
	Lemma~\ref{lem:attr}(\ref{sublem:complattr}) the set $S^j$ is closed for player~1 in
	$\game^j$, that is, player~2 has a strategy that keeps the play within $\game^j[S_j]$
	against any strategy of player~1. Since $S^j \cap \target^j_\ell = \emptyset$,
	this strategy is winning for player~2 (i.e., it satisfies $\objsty{coBüchi}{\target^j_\ell}$ and thus the disjunctive
	co-Büchi objective).
\end{proof}
}

{\subsection{Our Improved Algorithm}\label{sec:genBuchifast}}
\newcommand{\algConjBuchi}{
	\SetAlgoRefName{GenBuchiGame}
	\caption{Generalized Büchi Games in $O(k \cdot n^2)$ Time}
	\label{alg:ConjBuchi}
	\SetKwInOut{Input}{Input}
	\SetKwInOut{Output}{Output}
	\SetKw{break}{break}
	{\BlankLine}
	\Input{Game graph $\game = ((V, E),(\vo,\vt))$ and objective 
	$\bigwedge_{1 \le \ell \le k} \objsty{Büchi}{\target_\ell}$}
	\Output
	{
	  Winning set of player~1
	}
	{
	\BlankLine
	$\game^1 \gets \game$\;
	$\set{\target^1_\ell} \gets \set{\target_\ell}$\; 
	$j \gets 0$\;	
	}
		
	\Repeat
	 {
		$D^j = \emptyset$
	 }
	 {
		$j \gets j + 1$\;
	    \For
	      {
		$i \gets 1$ \KwTo $\lceil \log_2 n \rceil$
	      }
	      {
	      construct $G^j_i$\;
		$Z^j_i \gets \{v \in \vt \mid \OutDeg(G^j_i, v)=0\} \cup 
			     \{v \in \vo \mid \OutDeg(G^j, v)>2^i\}$\;
			\For{$1\leq \ell \leq k$}
			{
				$Y^j_{\ell,i} \gets \ato(G^j_i, \target^j_\ell \cup Z^j_i)$\;
				$S^j \gets V^j \setminus Y^j_{\ell,i}$\;
				\lIf{$S^j \ne \emptyset$}{player~2 dominion found, continue with line~\ref{l:domattr}}
			}
	      }
	    $D^j \gets \att(G^j, S^j)$\label{l:domattr}\;
	    $\game^{j+1} \gets \game^j \setminus D^j${\;}	    $\set{\target^{j+1}_\ell} \gets \set{\target^j_\ell \setminus D^j}$\; 	    
	  }
	  \Return $V^j$\;
}
The $O(k \cdot n^2)$-time Algorithm~\ref{alg:ConjBuchi} for generalized Büchi games combines
the basic algorithm{ for generalized Büchi games} described above
with the method used for the $O(n^2)$-time Büchi game 
algorithm~\cite{ChatterjeeH14}, called \emph{hierarchical graph 
decomposition}~\cite{HenzingerKW99}. The hierarchical
graph decomposition defines for a directed graph $G = (V, E)$ and integers 
$1 \le i \le \lceil \log_2 n \rceil$ the graphs $G_i = (V, E_i)$. 
Assume the incoming edges of each vertex in $G$ are given in some fixed 
order in which first the edges from vertices of $V_2$ and then the edges from
vertices of $V_1$ are listed. The set of edges $E_i$ contains all the 
outgoing edges of each $v \in V$ with $\OutDeg(G, v) \le 2^i$ and the first $2^i$
incoming edges of each vertex. Note that $G = G_{\lceil \log_2 n \rceil}$ and
$\lvert E_i \rvert \in O(n \cdot 2^i)$.{ See Figure~\ref{fig:decomp} for an example. 
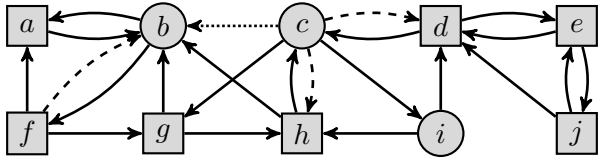
\begin{figure}
\centering
 % example.tex
\begin{tikzpicture}
\matrix[column sep=12mm, row sep=8mm]{
	\node[p2] (a) {$a$};
	& \node[p1] (b) {$b$};
	& \node[p1] (c) {$c$};
	& \node[p2] (d) {$d$};
	& \node[p2] (e) {$e$};\\
	\node[p2] (f) {$f$};
	& \node[p2] (g) {$g$};
	& \node[p2] (h) {$h$};
	& \node[p1] (i) {$i$};
	& \node[p2] (j) {$j$};\\
};
\path (a) edge[arrow, bend right] (b)
			(b) edge[arrow, bend right] (a)
					edge[arrow, bend left] (f)
			(c) edge[arrow, densely dotted] (b)
					edge[arrow] (g)
					edge[arrow, dashed, bend left] (h)
					edge[arrow] (i)
					edge[arrow, dashed, bend left] (d)
			(d) edge[arrow, bend left] (c)
					edge[arrow, bend left] (e)
			(e) edge[arrow, bend left] (d)
					edge[arrow, bend left] (j)
			(f) edge[arrow] (a)
					edge[arrow, bend left, dashed] (b)
					edge[arrow] (g)
			(g) edge[arrow] (b)
					edge[arrow] (h)
			(h) edge[arrow] (b)
					edge[arrow, bend left] (c)
			(i) edge[arrow] (h)
					edge[arrow] (d)
			(j)	edge[arrow] (d)
					edge[arrow, bend left] (e);
\end{tikzpicture}

\caption{Illustration of the hierarchical graph decomposition. 
Circles denote player~1 vertices, squares denote player~2 vertices. 
In the original graph $G$ all the edges (solid, dashed, dotted) are present.
From the hierarchical graph decomposition we consider the graphs $G_1$, $G_2$ and $G_3$.
$G_1$ contains only the solid edges, 
$G_2$ additionally the dashed edges, 
and $G_3$ also contains the dotted edge, i.e., $G=G_3$. 
Let the target sets be $T_1 = \set{a, e, i}$ and $T_2 = \set{b, d}$. 
Then the set $\set{e, j}$ is a player-2 dominion of $G$ that can be detected in $G_1$.
Its player-2 attractor contains additionally the vertex $d$. To detect that 
the remaining vertices are winning for player~1 it is necessary to consider the 
dotted edge.}
\label{fig:decomp}
\end{figure}} The runtime analysis uses that we can identify 
small player-2 dominions (i.e., player-1 closed sets that do not intersect
one of the target sets) that contain $O(2^i)$ vertices by considering only $G_i$.
The algorithm first searches for such a set $S^j$ in $G_i$ for $i = 1$ and 
each target set and then increases $i$ until the search is successful. 
In this way the time spent for the search is proportional to $k \cdot n$ times
the number of vertices in the found dominion, which yields a total runtime bound
of $O( k \cdot n^2)$. To obtain the $O(k \cdot n^2)$ running time bound, it 
is crucial to put the loop over the different Büchi sets as the innermost 
part of the algorithm.
Given a game graph $\game=(G,(V_1, V_2))$, we denote by $\game_i$ the game graph 
where $G$ was replaced by $G_i$ from the hierarchical graph decomposition, i.e., $\game_i=(G_i,(V_1, V_2))$.
{\begin{algorithm}[t]
	 \algConjBuchi 
	 \end{algorithm}}

{
\smallskip\par\noindent{\em Properties of hierarchical graph decomposition.}
The essential properties of the hierarchical graph decomposition for 
(generalized) Büchi games are summarized in the following lemma. The first part 
is crucial for correctness: When searching in $G_i$ for a player~1 closed set 
that does not contain one of the target sets, we can ensure that such a set 
is also closed for player~1 in $G$ by excluding certain vertices that
are missing outgoing edges in $G_i$ from the search.
The second part is crucial for the runtime argument: Whenever the basic 
algorithm would remove (i.e., identify as winning for player~2)
a set of vertices with at most $2^i$ vertices, then
we can identify this set also by searching in $G_i$ instead of~$G$.
The vertices~$Z_i$ we exclude for the search on $G_i$ are player-1 vertices
with more than $2^i$ outgoing edges and player-2 vertices with no 
outgoing edges in $G_i$. Note that the latter can only happen
if $\OutDeg(G, v) > 2^i$.
}
\begin{lemma}\label{lem:decomp}
	Let $\game = (G = (V, E), (V_1, V_2))$ be a game graph and $\set{G_i}$ its 
	hierarchical graph decomposition.
	For $1 \le i \le \lceil \log_2 n \rceil$ let 
	$Z_i$ be the set consisting of the player~2 vertices that have no outgoing edge
	in $\game_i$ and the player~1 vertices with $>2^i$ outgoing edges in~$\game$. 
	\begin{compactenum}
		\item If a set $S \subseteq V \setminus Z_i$ 
		is closed for player~1 in $\game_i$, then $S$ is closed for player~1 in $\game$.
		\label{sublem:sound}
		\item If a set $S \subseteq V$ is 
		closed for player~1 in $\game$ and $\lvert \att(\game, S) \rvert \le 2^i$,
		then \upbr{i} $\game_i[S] = \game[S]$, \upbr{ii} the set $S$ is in $V \setminus Z_i$, and
	  \upbr{iii} $S$ is closed for player~1 in $\game_i$.
		\label{sublem:size}
	\end{compactenum}
\end{lemma}
{
\begin{proof}
	We show the two points separately.
	\begin{enumerate}
		\item By $S \subseteq V \setminus Z_i$ we have for all $v \in S \cap V_1$
		that $\Out(G, v) = \Out(G_i, v)$.
		Thus if $\Out(G_i, v) \subseteq S$, then also $\Out(G, v) \subseteq S$.
		Each edge of $G_i$ is contained in $G$, thus we have for all $v \in S \cap V_2$
		that $\Out(G_i, v) \cap S \ne \emptyset$ implies $\Out(G, v) \cap S \ne \emptyset$.
		\item Since $S$ is closed for player~1 and $\lvert S \rvert \le 2^i$,
		\upbr{a} 
		the set $S$ does not contain vertices $v \in V_1$ with $\OutDeg(G, v) > 2^i$.
		Further for every vertex of $S$ also
		the vertices in $V_2$ from which it has incoming edges are contained 
		in $\att(\game, S)$. Thus by $\lvert \att(\game, S) \rvert \le 2^i$ no 
		vertex of $S$ has more than $2^i$ incoming edges from vertices of $V_2$.
		Hence, by the ordering of incoming edges in the construction of $G_i$, we obtain
		\upbr{b} for the vertices of $S$ all incoming edges from vertices of $V_2$
		are contained in $E_i$. 
		Combining \upbr{a}, i.e., $\Out(G, v) = \Out(G_i, v)$ for $v \in S \cap V_1$,
		and \upbr{b}, i.e., $(u, w) \in E_i$ for $u \in V_2$ and $w \in S$, 
		we have \upbr{i} $\game_i[S] = \game[S]$.
		Since $S$ is closed for player~1 in $\game$, every vertex
		$u \in S \cap V_2$ has an outgoing edge to another vertex $w \in S$ in $\game$. 
		Thus in particular these edges $(u, w)$ are contained in $E_i$ and 
		hence every vertex $u \in S \cap V_2$ has an outgoing edge to another
		vertex $w \in S$ in $G_i$. It follows that \upbr{ii} $S \cap Z_i = \emptyset$,
		and \upbr{iii} $S$ is closed for player~1 in $\game_i$ (by (1)).\qedhere
	\end{enumerate}
\end{proof}
}
{
That is, in all but the last iteration of Algorithm~\ref{alg:ConjBuchi} whenever the graph $G_i$
is considered a dominion of size at least $2^{i-1}$ is identified and removed from the graph.

\begin{corollary}\label{cor:size}
	If in Algorithm~\ref{alg:ConjBuchi} for some $\ell$, $j$, and $i > 1$
	we have that 
	$S^j = V^j \setminus \ato(\game^j_i, \target^j_\ell \cup Z^j_i)$ is not empty
	but for $i-1$ the set $V^j \setminus \ato(\game^j_{i-1}, \target^j_\ell \cup Z^j_{i-1})$
	is empty, then $\lvert \att(\game^j, S^j) \rvert > 2^{i-1}$.
\end{corollary}
\begin{proof}
	By Lemma~\ref{lem:attr}(\ref{sublem:complattr}) $S^j$ is closed for player~1 in
	$\game^j_i$ and by Lemma~\ref{lem:decomp}(\ref{sublem:sound}) also in $\game^j$.
	Assume by contradiction that $\lvert \att(\game^j, S^j) \rvert \le 2^{i-1}$.
	Then by Lemma~\ref{lem:decomp}(\ref{sublem:size}) we have that $S^j \subseteq 
	V^j \setminus Z^j_{i-1}$ and $S^j$ is closed for player~1 in $\game^j_{i-1}$.
	Since this means that in $\game^j_{i-1}$ 
	player~1 has a strategy to keep a play within $S^j$ against any 
	strategy of player~2
	and $S^j$ does not contain a vertex 
	of $Z^j_{i-1}$ or $\target^j_\ell$, the set $S^j$ does not intersect 
	with $\ato(\game^j_{i-1}, \target^j_\ell \cup Z^j_{i-1})$, a contradiction to
	$V^j \setminus \ato(\game^j_{i-1}, \target^j_\ell \cup Z^j_{i-1})$ being empty.
\end{proof}
} 

{
We next two Propositions show the correctness of the algorithm by 
(i) showing that all vertices in the final set $V^j$ are winning for player~1 and
(ii) all vertices not in $V^j$ are winning for player~2.
\begin{proposition}[Soundness Algorithm~\ref{alg:ConjBuchi}]
Let $V^{j^*}$ be the set returned by the algorithm.
Each vertex in $V^{j^*}$ is winning for player~1.
\end{proposition}

\begin{proof}
When the algorithm terminates we have $i = \lceil \log_2 n \rceil$ and 
$S^j = \emptyset$. Since for $i = \lceil \log_2 n \rceil$ 
we have $G^j_i = G^j$ and $Z^j_i = \emptyset$, 
the winning strategy of player~1 can be constructed in the same way 
as for the set returned by Algorithm~\ref{alg:ConjBuchiBasic}
(cf.\ Proof of Proposition~\ref{prop:soundness_ConjBuchiBasic}).
\end{proof}

\begin{proposition}[Completeness Algorithm~\ref{alg:ConjBuchi}]
	Let $V^{j^*}$ be the set returned by the algorithm.
	Player~2 has a winning strategy from each vertex in $V \setminus V^{j^*}$.
\end{proposition}

\begin{proof}
	By Lemma~\ref{lem:attr}(\ref{sublem:subgraph}) it is sufficient to show
	that, in each iteration $j$, player~2 has a winning strategy in $\game^j$ from each vertex of $S^j$.
	For a fixed $j$ with $S^j \ne \emptyset$, let $i$ and $\ell$ be such that
	$S^j = V^j \setminus \ato(\game^j_i, \target^j_\ell \cup Z^j_i)$. By 
	Lemma~\ref{lem:attr}(\ref{sublem:complattr}) the set $S^j$ is closed for player~1 in
	$\game^j_i$ and by Lemma~\ref{lem:decomp}(\ref{sublem:sound}) also in $\game^j$.
	That is, player~2 has a strategy that keeps the play within $\game^j[S_j]$
	against any strategy of player~1. Since $S^j \cap \target^j_\ell = \emptyset$,
	this strategy is winning for player~2 (i.e., satisfies the disjunctive
	co-Büchi objective).
\end{proof}

Finally, the $O(k \cdot n^2)$  runtime bound is by Corollary~\ref{cor:size}, Lemma~\ref{lem:attr}(\ref{sublem:attrtime}) and 
the fact that we can construct the graphs $G_i$ efficiently.

\begin{proposition}[Runtime Algorithm~\ref{alg:ConjBuchi}]\label{prop:time_conjbuchi}
	The algorithm can be implemented to terminate in $O(k \cdot n^2)$ time.
\end{proposition}

\begin{proof}
	To efficiently construct the graphs $G^j_i$ and the vertex sets $Z^j_i$ we
	maintain (sorted) lists of the incoming and the outgoing edges of each vertex.
	These lists can be updated whenever an obsolete entry is encountered in the 
	construction of $G^j_i$; as
	each entry is removed at most once, maintaining this data structures takes
	total time $O(m)$. For a given iteration~$j$ of the outer repeat-until loop and the $i$th 
	iteration of the inner repeat-until loop we have that 
	the graph $G^j_i$ contains $O(2^i \cdot n)$ edges and both $G^j_i$ and the 
	set $Z^j_i$ can be constructed from the maintained lists in time $O(2^i \cdot n)$.
	Further the $k$ attractor computations in the for-loop can be done in time 
	$O(k \cdot 2^i \cdot n)$, thus for any $j$ the $i$th iteration
	of the inner repeat-until loop takes time $O(k \cdot 2^i \cdot n)$.
	The time spent in the iterations up to the $i$th iteration forms a 
	geometric series and can thus also be bounded by $O(k \cdot 2^i \cdot n)$.
	When a non-empty set~$S^j$ is found in the $j$th iteration of the 
	outer repeat-until and in the $i$th iteration of the inner repeat-until loop, then by 
	Corollary~\ref{cor:size} we have $\lvert \att(\game^j, S^j) \rvert > 2^{i-1}$.
	The vertices in $\att(\game^j, S^j)$ are then removed from $G^j$ to obtain $G^{j+1}$
	and are not considered further by the algorithm. Thus we can charge the
	time of $O(k \cdot 2^i \cdot n)$ to identify $S^j$ to the vertices 
	in $\att(\game^j, S^j)$, which yields a bound on the total time spent in the 
	inner repeat-until loop, whenever $S^j \ne \emptyset$, of $O(k \cdot n^2)$.
	{By Lemma~\ref{lem:attr}(\ref{sublem:attrtime}) t}he total time for computing
	the attractors $\att(\game^j, S^j)$ can be bounded by $O(m)$.
	Finally the time for the last iteration of the while loop, when $S^j = 
	\emptyset$ and $i = \lceil \log_2 n \rceil$, can be bounded by $O(k \cdot 
	2^{\lceil \log_2 n \rceil} \cdot n) = O(k \cdot n^2)$.
\end{proof}
} 
{
\begin{remark}
Algorithm~\ref{alg:ConjBuchi} can easily be modified to also return winning strategies
for both players within the same time bound: For player~2 a winning strategy for 
the dominion~$D^j$ that is identified in iteration~$j$ of the algorithm 
can be constructed by combining his strategy to stay within the set~$S^j$ that is 
closed for player~1 with his attractor strategy to the set $S^j$.
For player~1 we can obtain a winning strategy by combining her attractor 
strategies to the sets $T_\ell$ for $1\le \ell \le k${ 
\upbr{as described in the proof of Proposition~\ref{prop:soundness_ConjBuchiBasic}}}.
\end{remark}
}

 % ConjBuchiLB
\section{Conditional Lower bounds for Generalized Büchi Games}\label{sec:lowerbounds}

In this section we present two conditional lower bounds, one for dense graphs
($m = \Theta(n^2)$) based on STC \& BMM, and one for sparse graphs 
($m = \Theta(n^{1+o(1)})$) based on OVC \& SETH. 

\begin{theorem}\label{thm:genbuchi_STChard}
  There is no combinatorial $O(n^{3-\epsilon})$ or $O((k\cdot n^2)^{1-\epsilon})$-time algorithm \upbr{for any $\epsilon > 0$} for generalized Büchi games under  Conjecture~\ref{conj:triangle} \upbr{i.e., unless STC \& BMM fail}. 
  {In particular, there is no such algorithm deciding whether the winning set is non-empty
  or deciding whether a specific vertex is in the winning set.}
\end{theorem}

The result can be obtained from a reduction from triangle detection to disjunctive co-B\"uchi objectives on graphs in~\cite{ChatterjeeDHL16}, 
and we present the reduction in terms of game graphs below and illustrate it on an example in Figure~\ref{fig:TriangletoGraphs}.
\begin{reduction}\label{red:TriangletoGraphs}
 Given a graph $G=(V,E)$ \upbr{for triangle detection}, we build a game graph $\game'=((V',E'),(V_1,V_2))$ 
 \upbr{for generalized Büchi objectives} as follows.
 As vertices $V'$ we have four copies $V^1, V^2, V^3, V^4$ of $V$ and a vertex $s$.
 A vertex $v^i \in V^i, i \in \{1,2,3\}$ has an edge to a vertex $u^{i+1} \in V^{i+1}$ iff $(v,u) \in E$. 	  
 Moreover, $s$ has an edge to all vertices of $V^1$ and all vertices of $V^4$ have an edge to~$s$.
 All the vertices are owned by player 2, i.e., $V_1=\emptyset$ and $V_2=V$.
 Finally, we consider the generalized Büchi objective $\bigwedge_{v \in V} \objsty{Büchi}{\target_v}$, with
 $\target_v= (V^1 \setminus \{v^1\}) \cup (V^4 \setminus \{v^4\})$.
\end{reduction}
{The game graph $\game'$ is constructed such that  there is a triangle in the graph $G$ if and only if 
the vertex $s$ is winning for player~2 in the generalized Büchi game on $\game'$.
For instance consider the example in Figure~\ref{fig:TriangletoGraphs}. The graph $G$ has a triangle $a$,$b$,$c$
and this triangle gives rise to the following winning strategy for player~2 starting at vertex $s$. 
When a play is in the vertex $s$ then player~2 moves to vertex $a^1$, when in $a^1$ he moves to $b^2$,
when in $b^2$ he moves to $c^3$, when in $c^3$ he moves to $a^4$, and finally from $a^4$ he moves back to $s$.
This strategy does not visit any vertex of the set $T_a$ and thus the conjunctive Büchi objective of player~1 is not satisfied, 
i.e., player~2 wins. 
The following Lemma we show that also the other direction holds, i.e., that a memoryless winning strategy from $s$
gives rise to a triangle in the original graph. 
This correspondence between triangles and memoryless winning strategies then gives the correctness of the reduction.}
{
\begin{figure}
 \centering
  \begin{tikzpicture}[yscale=0.8, xscale=1.6,>=stealth]
   \small
		\draw  (0,-2.25)node[player2](s){$s$};
  		\path 	(1,0)node[player2](a1){$a^1$}
			++(0,-1.5)node[player2](b1){$b^1$}
			++(0,-1.5)node[player2](c1){$c^1$}
			++(0,-1.5)node[player2](d1){$d^1$}
			;
  		\path 	(2,0)node[player2](a2){$a^2$}
			++(0,-1.5)node[player2](b2){$b^2$}
			++(0,-1.5)node[player2](c2){$c^2$}
			++(0,-1.5)node[player2](d2){$d^2$}
			;
  		\path 	(3,0)node[player2](a3){$a^3$}
			++(0,-1.5)node[player2](b3){$b^3$}
			++(0,-1.5)node[player2](c3){$c^3$}
			++(0,-1.5)node[player2](d3){$d^3$}
			;
		\path 	(4,0)node[player2](a4){$a^4$}
			++(0,-1.5)node[player2](b4){$b^4$}
			++(0,-1.5)node[player2](c4){$c^4$}
			++(0,-1.5)node[player2](d4){$d^4$}
			;
		\path [->, thick]
			(s) edge (a1)
			(s) edge (b1)
			(s) edge (c1)
			(s) edge (d1)
			(a1) edge (b2)
			(a2) edge (b3)
			(a3) edge (b4)
			(b1) edge (c2)
			(b2) edge (c3)
			(b3) edge (c4)
			(c1) edge (a2)
			(c2) edge (a3)
			(c3) edge (a4)
			(b1) edge (a2)
			(b2) edge (a3)
			(b3) edge (a4)
			
			(c1) edge (d2)
			(c2) edge (d3)
			(c3) edge (d4)
			
			(d1) edge (a2)
			(d2) edge (a3)
			(d3) edge (a4)
			;
		\draw[left,thick,->,rounded corners=5pt]  
		(a4) -- ++(0.5,0) -- (4.5,-5.2) -- (0,-5.2) -- (s);
		\draw[left,thick,->,rounded corners=5pt]  
		(b4) -- ++(0.5,0) -- (4.5,-5.2) -- (0,-5.2) -- (s);
		\draw[left,thick,->,rounded corners=5pt]  
		(c4) -- ++(0.5,0) -- (4.5,-5.2) -- (0,-5.2) -- (s);
		\draw[left,thick,->,rounded corners=5pt]  
		(d4) -- ++(0.5,0) -- (4.5,-5.2) -- (0,-5.2) -- (s);

 \end{tikzpicture}
 \caption{Illustration of Reduction~\ref{red:TriangletoGraphs}, with $G=(\{a,b,c,d\},\{(a,b), (b,a),(b,c), (c,a),(c,d),$ $(d,a)\})$.
 The target sets for disjunctive co-Büchi are $T_a = \set{b^1, c^1, d^1, b^4, c^4, d^4}$, $T_b = \set{a^1, c^1, d^1, a^4, c^4, d^4}$, $T_c = \set{a^1, b^1, d^1, a^4, b^4, d^4}$, and $T_d = \set{a^1, b^1, c^1, a^4, b^4, c^4}$.
 }
 \label{fig:TriangletoGraphs}
\end{figure}
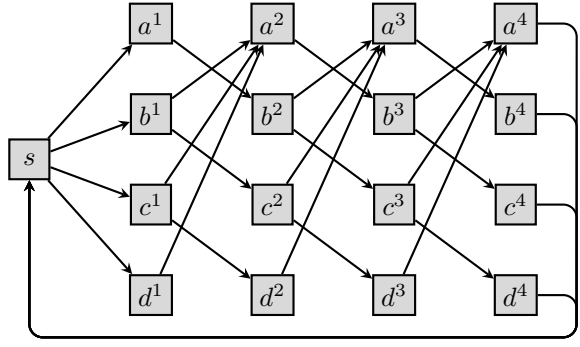
} 

{
\begin{lemma}
 Let $\game'$ be the game
 graph given by Reduction~\ref{red:TriangletoGraphs} for a graph~$G$
 and let $\target_v= (V^1 \setminus \{v^1\}) \cup (V^4 \setminus \{v^4\})$. Then 
 {the following statements are equivalent.
 \begin{enumerate}
    \item $G$ has a triangle.
    \item $s \not\in W_1(\game',\bigwedge_{v \in V} \objsty{Büchi}{\target_v})$.
    \item The winning set $W_1(\game', \bigwedge_{v \in V} \objsty{Büchi}{\target_v})$ is empty.
 \end{enumerate}
 }
\end{lemma}

\begin{proof}
 (1)$\Rightarrow$(2): Assume that
 $G$ has a triangle with vertices $a,b,c$ and 
 let $a^i$, $b^i$, $c^i$ be the copies of $a,b,c$ in $V^i$.
 Now a strategy for player~2 in $\game'$ to satisfy $\objsty{coBüchi}{\target_a}$ is as follows:
 When in $s$, go to $a^1$; when in $a^1$, go to $b^2$; when in $b^2$, go to $c^3$;
 when in $c^3$, go to $a^4$; and when in $a^4$, go to $s$.
 As $a,b,c$ form a triangle, all the edges required by the above strategy exist.
 When player~2 starts at~$s$ and follows the above strategy,
 then he plays an infinite path that only uses the vertices $s,a^1,b^2,c^3,a^4$ and 
 thus satisfies $\objsty{coBüchi}{\target_a}$.

 (2)$\Rightarrow$(1): Assume that there is a memoryless winning strategy for player~2
 starting in $s$ and satisfying $\objsty{coBüchi}{\target_a}$.
 Starting from $s$, this strategy has to go to $a^1$, as all other successors of $s$ are
 contained in $\target_a$ and thus would violate the $\objsty{coBüchi}{\target_a}$ objective.
 Then the play continues on some vertex $b^2 \in V^2$ and $c^3 \in V^3$
 and then, again by the coBüchi constraint, has to enter $a^4$.
 Now by construction of $\game'$ we know that there must be edges $(a,b), (b,c), (c,a)$ in the original graph $G$,
 i.e.\ there is a triangle in $G$. 
 
 (2)$\Leftrightarrow$(3): 
 Notice that when removing $s$ from $\game'$ we get an acyclic graph and thus each infinite
 path has to contain $s$ infinitely often. Thus, if the winning set is non-empty,
 there is a cycle winning for some vertex and then 
 this cycle is also winning for $s$. For the converse direction we have that if $s$ is in the winning set, 
 then the winning set is non-empty.
\end{proof}

The size and the construction time of the game graph $\game'$, given in Reduction~\ref{red:TriangletoGraphs}, 
are linear in the size of the original graph $G$ and we have $k = \Theta(n)$ target sets.
Thus if we would have a combinatorial $O(n^{3-\epsilon})$  or $O((k\cdot n^2)^{1-\epsilon})$ algorithm for generalized Büchi games, 
we would immediately get a combinatorial  $O(n^{3-\epsilon})$ algorithm for triangle detection, 
which contradicts STC (and thus BMM).

Notice that the sets $\target_v$ in the above reduction are of linear size but 
can be reduced to logarithmic size by modifying the graph constructed in Reduction~\ref{red:TriangletoGraphs} as follows.
Remove all edges incident to $s$ and  replace them by two complete binary trees. 
The first tree with $s$ as root and the vertices $V^1$ as leaves is directed towards the leaves, 
the second tree with root~$s$ and leaves~$V^4$ is directed towards~$s$.
Now for each pair $v^1,v^4$ one can select one vertex of each non-root level of the trees 
to be in the set $\target_v$ such that the only winning path for player~2 starting in $s$ has to use 
$v^1$ and each winning path for player~2 to $s$ must pass $v^4$ (see also \cite{ChatterjeeDHL16}).
\smallskip 
}
{\par Next we present an $\Omega(m^{2-o(1)})$ lower bound for generalized Büchi objectives in sparse game graphs based on OVC and SETH.}
\begin{theorem}\label{thm:genbuchi_OVChard}
  There is no $O(m^{2-\epsilon})$ or $O(\min_{1\!\leq\,i\,\leq\!k} b_i \cdot (k \cdot m)^{1-\epsilon})$-time algorithm \upbr{for any $\epsilon\!>\!0$} 
  for generalized Büchi games under Conjecture~\ref{conj:ov} \upbr{i.e., unless
  OVC{ and }SETH fail}.
  {In particular, there is no such algorithm for deciding whether the winning set is non-empty
  or deciding whether a specific vertex is in the winning set.}
\end{theorem}

{The above theorem is by a linear time reduction from OV provided below,
and illustrated on an example in Figure~\ref{fig:OVtoConjBuchiGraphGames}.}

\begin{reduction}\label{red:OVtoConjBuchiGraphGames}
 Given two sets $S_1, S_2$ of $d$-dimensional vectors, we build the following game graph $\game = ((V, E), (V_1, V_2))$.
 {
 \begin{itemize}
  \item   The vertices $V$ 
  are given by a start vertex~$s$, sets of vertices $S_1$ and $S_2$ representing the 
  sets of vectors, and a set of
  vertices $\mathcal{C}=\{c_i \mid 1 \leq i \leq d\}$ representing the 
  coordinates. The edges $E$ are defined as follows: the start vertex~$s$
  has an edge to every vertex of $S_1$ and every vertex of $S_2$ has an edge to $s$;
  further for each $x \in S_1$
  there is an edge {to} $c_i \in \mathcal{C}$ iff $x_i=1$ and for each $y \in S_2$
  there is an edge {from} $c_i \in \mathcal{C}$ iff $y_i=1$.
  \item The set of vertices $V$ is partitioned into player~1 vertices $V_1= S_1 \cup 
  S_2 \cup \mathcal{C}$ and player~2 vertices $V_2=\{s\}$.
 \end{itemize}}
   Finally, the generalized Büchi objective is given by $\bigwedge_{v \in S_2} \objsty{Büchi}{\target_v}$ with $\target_v=\{v\}$.
\end{reduction}
{
  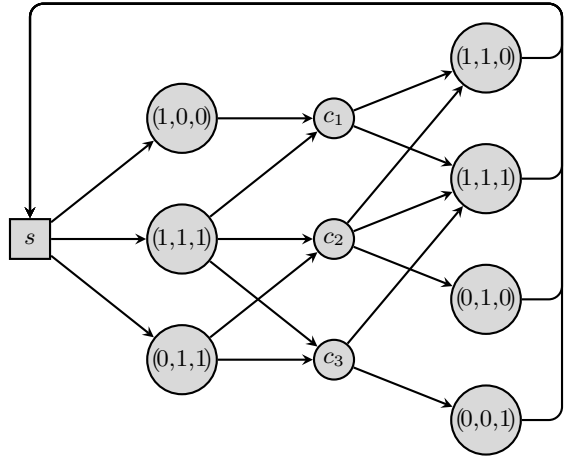
\begin{figure}
  \centering
  \begin{tikzpicture}[yscale=0.8,>=stealth]
      \footnotesize
 		\path 	node[player2](s){$s$}
			++(2,2) node[vplayer1](x1){$(\!1,\!0,\!0\!)$}
			++(0,-2)node[vplayer1] (x2){$(\!1,\!1,\!1\!)$}
			++(0,-2) node[vplayer1](x3){$(\!0,\!1,\!1\!)$}
			(4,2) node[player1](c1){$c_1$}
			++(0,-2) node[player1](c2){$c_2$}
			++(0,-2) node[player1](c3){$c_3$}
			(6,3) node[vplayer1](y1){$(\!1,\!1,\!0\!)$}
			++(0,-2) node[vplayer1](y2){$(\!1,\!1,\!1\!)$}
			++(0,-2)node[vplayer1] (y3){$(\!0,\!1,\!0\!)$}
			++(0,-2) node[vplayer1](y4){$(\!0,\!0,\!1\!)$}
			;
		\normalsize
		\path [->, thick]
			(s) edge (x1)
			(s) edge (x2)
			(s) edge (x3)
			(x1) edge (c1)
			(x2) edge (c1)
			(x2) edge (c2)
			(x2) edge (c3)
			(x3) edge (c2)
			(x3) edge (c3)
			[<-]
			(y1) edge (c1)
			(y1) edge (c2)
			(y2) edge (c1)
			(y2) edge (c2)
			(y2) edge (c3)
			(y3) edge (c2)
			(y4) edge (c3)
			;
		\draw[left,thick,->,rounded corners=5pt]  (y1) -- ++(1,0) -- (7,3.9) -- (0,3.9) -- (s);
		\draw[left,thick,->,rounded corners=5pt]  (y2) -- ++(1,0) -- (7,3.9) -- (0,3.9) -- (s);
		\draw[left,thick,->,rounded corners=5pt]  (y3) -- ++(1,0) -- (7,3.9) -- (0,3.9) -- (s);
		\draw[left,thick,->,rounded corners=5pt]  (y4) -- ++(1,0) -- (7,3.9) -- (0,3.9) -- (s);
  \end{tikzpicture} 
  \caption{Illustration of Reduction~\ref{red:OVtoConjBuchiGraphGames}, for $S_1=\{(1,0,0), (1,1,1), (0,1,1)\}$ and 
    $S_2=\{(1,1,0), (1,1,1), (0,1,0), (0,0,1)\}$.}
  \label{fig:OVtoConjBuchiGraphGames}
  \end{figure}
  } 

{The correctness of the reduction is by the following lemma.}

\begin{lemma}\label{lem:OVtoConjBuchiGraphGames}
    {Given two sets $S_1, S_2$ of $d$-dimensional vectors,
  the corresponding graph game~$\game$ given by
  Reduction~\ref{red:OVtoConjBuchiGraphGames}, and  
  $\target_v=\{v\}$ for $v \in S_2$,
  the following statements are equivalent:
  \begin{enumerate}
    \item There exist orthogonal vectors $x \in S_1$ and $y \in S_2$.
    \item $s \not\in W_1(\game, \bigwedge_{v \in S_2} \objsty{Büchi}{\target_v})$
    \item The winning set $W_1(\game, \bigwedge_{v \in S_2} \objsty{Büchi}{\target_v})$ is empty.
  \end{enumerate}}
  \end{lemma}
\begin{proof}
 W.l.o.g.\ we assume that the $1$-vector, i.e., the vector with all coordinates being $1$, is  contained in $S_2$
 (adding the $1$-vector does not change the result of the OV instance), which guarantees that each vertex $c \in \mathcal{C}$ has at least 1 outgoing edge.
 Then a play in the game graph~$\game$ proceeds as follows. 
 Starting from $s$, player~2 chooses a vertex $x \in S_1$; 
 then player~1 first picks a vertex $c \in \mathcal{C}$  and then a vertex $y \in S_2$; 
 then the play goes back to $s$ (at each $y \in S_2$ player~1 has only this choice), 
 starting another cycle of the play.
 {\par}
 (1)$\Rightarrow$(2): Assume there are orthogonal vectors $x \in S_1$ and $y \in S_2$.
    Now player~2 can satisfy $\objsty{coBüchi}{\target_y}$ by simply going to $x$ 
    whenever the play is in $s$.
    Then player~1 can choose some adjacent $c \in \mathcal{C}$
    and then some adjacent vertex in $S_2$, 
    but as $x$ and $y$ are orthogonal, this $c$ is not connected to $y$. 
    Thus the play will never visit $y$.
 {\par}
 (2)$\Rightarrow$(1): 
  {By the fact that generalized Büchi games satisfy Determinacy, i.e., 
    $W_1=V \setminus W_2$ (cf.\ Theorem~\ref{th:det}), 
    we have that (2) is equivalent to  $s \in W_2(\game, \bigwedge_{v \in S_2} \objsty{Büchi}{\target_v})$.}
     Assume $s \in W_2(\game, \bigwedge_{v \in S_2} \objsty{Büchi}{\target_v})$ and consider
    a corresponding strategy for player~2 that satisfies $\bigvee_{v \in S_2} \objsty{coBüchi}{\target_v}$. 
    Notice that the graph is such that player~2 has to visit at least one of the vertices $v$ in $S_1$
    infinitely often.
    Moreover, for such a vertex $v$ then player~1 can visit all vertices $v' \in S_2$ that correspond to vectors not orthogonal to~$v$ infinitely often. 
    That is, if $v$ has no orthogonal vector, player~1 can satisfy all the Büchi constraints,
    a contradiction to our assumption that $s \in W_2(\game, \bigwedge_{v \in S_2} \objsty{Büchi}{\target_v})$.
    Thus there must be a vector $x\in S_1$ such that there exists a vector $y \in S_2$ that is orthogonal to~$x$.
  {\par
  (2) $\Leftrightarrow$ (3):
     Notice that when removing $s$ from the graph we get an acyclic graph and 
     thus each infinite  path has to contain $s$ infinitely often. 
     Certainly if $s$ is in the winning set of player~1, this set is non-empty. 
     Thus let us assume there is a vertex $v$ different from $s$ with a winning strategy $\sigma$.
     All (winning) paths starting in $v$ cross $s$ after at most $3$ steps and thus 
     the strategy~$\sigma$ is also winning for player~1 when the play starts at~$s$.}
\end{proof}
Let $N=\max(|S_1|,|S_2|)$.
The number of vertices in the game graph, constructed by Reduction~\ref{red:OVtoConjBuchiGraphGames}, is $O(N)$, 
the number of edges~$m$ is $O(N \log N)$ (recall that $d \in O(\log N)$), 
we have $k \in \Theta(N)$ target sets, each of size $1$, and the construction can be performed in $O(N \log N)$ time.
Thus, if we would have an $O(m^{2-\epsilon})$ or $O(\min_{1 \leq i\leq k} b_i \cdot (k \cdot m)^{1-\epsilon})$ time
algorithm for any $\epsilon > 0$, we
would immediately get an  $O(N^{2-\epsilon})$ algorithm for OV, which contradicts OVC (and thus SETH).
{
\begin{remark}
Notice that the conditional lower bounds apply to instances with $k \in \Theta(n^c)$ for arbitrary $0<c\leq 1$, 
although the reductions produce graphs with $k \in \Theta(n)$. 
The instances constructed by the reductions have the property that whenever player~2 has a winning strategy, 
he also has a winning strategy for a specific co-Büchi set $\target_v$. 
Now instead of solving the instance with $\Theta(n)$ many target sets, 
one can simply consider $O(n^{1-c})$ many instances with $\Theta(n^c)$ target sets 
and obtain the winning set for player~2 in the original instance by the union of the 
player~2 winning sets of the new instances. Finally, towards a contradiction, assume there would 
be an $O((k\cdot f(n,m))^{1-\epsilon})$-time algorithm for $k \in \Theta(n^c)$.
Then together we the above observation we would get an $O((k\cdot f(n,m))^{1-\epsilon})$-time 
algorithm for the original instance.
\end{remark}}
{\begin{remark}
In both reductions the constructed graph becomes acyclic when deleting vertex~$s$. Thus, our lower bounds
also apply for a broad range of digraph parameters. 
For instance let $w$ be the DAG-width~\cite{BerwangerDHK06} of a graph, then 
there is no $O(f(w) \cdot (k\cdot n^2)^{1-\epsilon})$-time algorithm \upbr{under BMM} and no $O(f(w)\cdot (k m)^{1-\epsilon})$-time algorithm \upbr{under SETH}.
\end{remark}}

 % GR1Alg
\newcommand{\pp}{\mbox{\small ++}}

\section{Generalized Reactivity-1 Games}\label{sec:GR1Alg}

GR(1) games deal with an objective of the form 
$\bigwedge_{t=1}^{k_1} \objsty{Büchi}{L_t} \rightarrow \bigwedge_{\ell=1}^{k_2} \objsty{Büchi}{U_\ell}$ and 
can be solved in $O(k_1 k_2 \cdot  m \cdot n)$ time~\cite{JuvekarP06}
with an extension of the progress measure algorithm of~\cite{Jurdzinski00}
and in $O((k_1 k_2 \cdot n)^{2.5})$ time by combining the reduction 
to one-pair Streett objectives by~\cite{BloemCGHJ10} with the algorithm 
of~\cite{ChatterjeeHL15}. {}
In this section we develop an $O(k_1 k_2 \cdot n^{2.5})$-time
algorithm by modifying the algorithm of~\cite{JuvekarP06} to compute dominions.
We further use our 
$O(k \cdot n^2)$-time algorithm for generalized Büchi games with $k = k_1$
as a subroutine.

{\smallskip\noindent\emph{Section Outline.}}
We first describe a basic, direct algorithm for GR(1) games that is based
on repeatedly identifying player-2 dominions in generalized Büchi games. 
We then show how the progress measure 
algorithm of~\cite{JuvekarP06} can be modified to identify 
player-2 dominions in generalized Büchi games with $k_1$ Büchi objectives 
in time proportional to $k_1 \cdot m$ times the size of the dominion. 
In the $O(k_1 k_2 \cdot n^{2.5})$-time algorithm we use the modified 
progress measure algorithm in combination with the hierarchical graph decomposition 
of~\cite{ChatterjeeH14,ChatterjeeHL15} to identify dominions that contain up to 
$\sqrt{n}$ vertices and our $O(k_1 \cdot n^2)$-time algorithm for generalized 
Büchi games to identify dominions with more than $\sqrt{n}$ vertices.
Each time we search for a dominion we might have to consider $k_2$ different
subgraphs.

{\smallskip\noindent\emph{Notation.} 
In the algorithms and their analysis we denote the sets in the} $j$th-iteration of our algorithms 
with superscript~$j$, in particular
$\game^1 = \game$, where $\game$ is the input game graph, 
$G^j$ is the graph of $\game^j$,
$V^j$ is the vertex set of $G^j$, $V_1^j =V_1 \cap V^j$, $V_2^j =V_2 \cap V^j$,
$L^j_t = L_t \cap V^j$, and $U^j_\ell = U_\ell \cap V^j$.

{\subsection{Basic Algorithm for GR(1) Objectives}}
Similar to generalized Büchi games, the basic algorithm for GR(1) games, described in Algorithm~\ref{alg:GR1GameBasic},
identifies a player-2 dominion~$S^j$, removes the dominion and its 
player-2 attractor~$D^j$
from the graph, and recurses on the remaining game graph~$\game^{j+1} = \game^j \setminus D^j$. If no player-2 dominion
is found, the remaining set of vertices~$V^j$ is returned as the winning set 
of player~1. Given the set $S^j$ is indeed a player-2 dominion,
the correctness of this approach follows from Lemma~\ref{lem:attr}(\ref{sublem:subgraph}).
A player-2 dominion in $\game^j$ is identified as follows:
For each $1 \le \ell \le k_2$ first the player-1 attractor~$Y^j_\ell$
of $U^j_\ell$ is temporarily removed from the graph. 
Then a generalized Büchi game with target sets $L^j_1, \ldots,
L^j_{k_1}$ is solved on $\overline{\game^j \setminus Y^j_\ell}$. 
The generalized Büchi player in this game corresponds
to player~2 in the GR(1) game and his winning set to a player-2 dominion
in the GR(1) game.
Note that $V^j \setminus Y^j_\ell$ is player-1 closed and does not contain
$U^j_\ell$. Thus if in the game induced by the vertices of $V^j \setminus Y^j_\ell$
player~2 can win w.r.t.\ the generalized
Büchi objective $\bigwedge_{t=1}^{k_1} \objstytxt{Büchi}{L^j_t}$, 
then these vertices form a player-2 dominion in the GR(1) game. 
{This observation is formalized in Lemma~\ref{lem:dominionGR1ConjBuchi}.} Further, we can 
show that when a player-2 dominion in the GR(1)
game on $\game^j$ exists, then for one of the sets $U^j_\ell$ the winning set 
of the generalized Büchi game on $\overline{\game^j \setminus Y^j_\ell}$ is non-empty;
otherwise we can construct a winning strategy of player~1 for the GR(1) game
on $\game^j${(see Lemma~\ref{lem:domwithoutU} and Proposition~\ref{prop:p1win_GR1Basic})}.
Note that this algorithm computes a player-2 dominion $O(k_2 \cdot n)$ often using our $O(k_1 \cdot n^2)$-time generalized Büchi Algorithm~\ref{alg:ConjBuchi}{ from Section~\ref{sec:genBuchifast}}.

\begin{theorem}\label{th:GR1GameBasic}
The basic algorithm for GR(1) games computes the winning set of player~1 
in $O(k_1 \cdot k_2 \cdot n^3)$~time.
\end{theorem}

{
We first show that the dominions we compute via the generalized Büchi games are indeed 
player-2 dominions for the GR(1) game.
\begin{lemma}\label{lem:dominionGR1ConjBuchi}
 We are given a game with game graph~$\game$ and GR(1) 
 objective $\bigwedge_{t=1}^{k_1}\objsty{Büchi}{L_t} \rightarrow \bigwedge_{\ell=1}^{k_2} \objsty{Büchi}{U_\ell}$.
  Each player-1 dominion $D$ of the game graph $\overline{\game}$ with generalized Büchi objective 
  $\bigwedge_{t=1}^{k_1} \objsty{Büchi}{L_t}$,  
  for which there is an index $1 \leq \ell \leq k_2$ with $D \cap U_\ell=\emptyset$,
  is a player-2 dominion of $\game$ with the original GR(1) objective.
\end{lemma}
\begin{proof}
 By definition of a dominion, in  $\overline{\game}$ player~1 has a strategy that visits all sets $L_t$ infinitely often and 
 only visits vertices in $D$.
 But then for some $\ell$ the Büchi set $U_\ell$ is not visited at all and thus in $\game$ the strategy is winning for player~2 w.r.t.\ the GR(1) objective.
\end{proof}

Next we show that each player-2 dominion contains a sub-dominion that does not intersect with 
one of the sets $U_\ell$, and thus can be computed via generalized Büchi games.

\begin{lemma}\label{lem:domwithoutU}
 We are given a game with game graph $\game$ and 
 GR(1) objective $\bigwedge_{t=1}^{k_1} \objsty{Büchi}{L_t} \rightarrow \bigwedge_{\ell=1}^{k_2} \objsty{Büchi}{U_\ell}$.
 Each player-2 dominion $D$ has a subset $D' \subseteq D$ that is a player-2 dominion
 with $D' \cap U_\ell=\emptyset$ for some $1 \leq \ell \leq k_2$.
\end{lemma}
\begin{proof}
First, note that $D$ is closed for player-1 and thus by the definition of a 2-dominion
we have that $D$ is equal to $W_2(\game[D],\bigwedge_{t=1}^{k_1} \objsty{Büchi}{L_t \cap D} \rightarrow \bigwedge_{\ell=1}^{k_2} \objsty{Büchi}{U_\ell \cap D})$.
Moreover, as each 1-closed set in $\game[D]$ is also 1-closed in $\game$, 
a set $D' \subseteq D$ is a player-2 dominion of $\game$ iff it is a player-2 dominion of $\game[D]$.

Towards a contradiction, assume that there does not exist such a player-2 dominion $D'$ in $\game[D]$. We will show that then player~1 can win from the vertices of~$D$.
By the assumption we have that player-1 has a winning strategy for each
$1 \le \ell \le k_2$ for the game graph
$
  \game^\ell[D] =\game[D] \setminus \ato(\game[D], U_\ell)
$
w.r.t.\ the GR(1) objective.
As $U_\ell \cap \game^\ell[D] =\emptyset$, the same strategy is also winning for
the disjunctive co-Büchi objective $\bigvee_{t=1}^{k_1} \objsty{coBüchi}{L_t}$.
Now consider the following strategy for player~1 in $\game[D]$.
The winning strategy of player~1 is constructed from her winning strategies for 
the game graphs $\game^\ell[D]$ and the attractor strategies for 
$\ato(\game[D], U_\ell)$ for $1 \le \ell \le k_2$ as follows.
Player~1 maintains a counter $c \in \set{1, \ldots, k_2}$ that is initialized to~1.
As long as the current vertex in the play is contained in $\game^c[D]$,
player~1 plays her winning strategy for $\game^c[D]$.
If a vertex of $\ato(\game[D], U_c)$ is reached, player~1
follows the corresponding attractor strategy until $U_c$ is reached. Then 
player~1 increases the counter by one or sets the counter to~1 if its value was
$k_2$ and continues playing the above strategy for the new value $c$.
In each  play one of two cases must happen:
\begin{itemize}
 \item Case~1: After some prefix of the play for some counter value $c$ the set 
$\ato(\game[D], U_c)$ is never reached. Then the play satisfies the 
disjunctive co-Büchi objective $\bigvee_{t=1}^{k_1} \objsty{coBüchi}{L_t}$ and 
thus the GR(1) objective.
 \item Case~2: For all $c \in \set{1, \ldots k_2}$
the set $U_c$ is reached infinitely often. Then the play satisfies the
generalized Büchi objective $\bigwedge_{\ell=1}^{k_2} \objsty{Büchi}{U_\ell}$ and 
thus the GR(1) objective.
\end{itemize}
Hence, we have shown that $D \subseteq W_1$, a contradiction.
\end{proof}

\begin{algorithm}[t]
		\SetAlgoRefName{GR(1)GameBasic}
	\caption{GR(1) Games in $O(k_1 \cdot k_2 \cdot n^3)$ Time} 
	\label{alg:GR1GameBasic}
	\SetKwInOut{Input}{Input}
	\SetKwInOut{Output}{Output}
	\SetKw{break}{break}
	\BlankLine
	\Input{Game graph $\game$, Obj.\ $\bigwedge_{t=1}^{k_1} \objsty{Büchi}{L_t} \rightarrow \bigwedge_{\ell=1}^{k_2} \objsty{Büchi}{U_\ell}$}
	\Output
	{
	  Winning set of player~1
	}
	\BlankLine
	$\game^1
	\gets \game$\;
	$\set{U^1_\ell} \gets \set{U_\ell}$; $\set{L^1_t} \gets \set{L_t}$\;
	$j \gets 0$\;
	\Repeat
	  {
	    $D^j= \emptyset$
	  }
	  {
	    $j \gets j+1$ \;
	    \For{$1 \leq \ell \leq k_2$}
	      {
		  $Y^j_\ell \gets \ato(\game^j, U^j_\ell)$\;
		  $S^j \gets W_1\left(\overline{\game^j \setminus Y^j_\ell}, 
		  \bigwedge_{t=1}^{k_1} \objsty{Büchi}{L^j_t  \setminus Y^j_\ell} \right)$\;
		  \lIf{$S^j \not= \emptyset$}{\break}
	      }
	     $D^j \gets \att(\game^j, S^j)$\;
	    $\game^{j+1} 
	    \gets \game^j \setminus D^j$\;
	    $\set{U^{j+1}_\ell} \gets \set{U^j_\ell \setminus D^j}$; 
	    $\set{L^{j+1}_t} \gets \set{L^j_t \setminus D^j}$\;  
	  }
	  \Return $V^j$
\end{algorithm}

Let the $j^*$th iteration be the last iteration of Algorithm~\ref{alg:GR1GameBasic}.
For the final game graph $\game^{j^*}$ we can build a winning strategy for player~1
by combining her winning strategies for the disjunctive objective in the subgraphs $\overline{\game_\ell^{j^*}}$
and the attractor strategies for $\ato(\game^{j^*}, U_\ell)$. 

\begin{proposition}[Soundness Algorithm~\ref{alg:GR1GameBasic}]\label{prop:p1win_GR1Basic}
Let $V^{j^*}$ be the set of vertices returned by Algorithm~\ref{alg:GR1GameBasic}.
Each vertex in $V^{j^*}$ is winning for player~1.
\end{proposition}
\begin{proof}
First note that $V^{j^*}$ is closed for player~2 by 
Lemma~\ref{lem:attr}(\ref{sublem:complattr}). 
Thus as long as player~1 plays a strategy that stays within $V^{j^*}$, a play
that reaches $V^{j^*}$ will never leave $V^{j^*}$. The following strategy for 
player~1 for the vertices of $V^{j^*}$ satisfies this condition.
The winning strategy of player~1 is constructed from the winning strategies 
of  the disjunctive co-Büchi player, i.e., player~2, in the generalized Büchi 
games with game graphs 
$\overline{\game_\ell^{j^*}} = \overline{\game^{j^*} \setminus Y^{j^*}_\ell}$ 
and objectives $\bigwedge_{t=1}^{k_1} \objsty{Büchi}{L^{j^*}_t \setminus Y^{j^*}_\ell}$
and the attractor strategies for $\ato(\game^{j^*}, U_\ell^{j^*})$ for $1 \le \ell \le k_2$.
Player~1 maintains a counter $c \in \set{1, \ldots, k_2}$ that is initialized to~1
and proceeds as follows.
(1) As long as the current vertex in the play is contained in $G^{j^*}_c = G^{j^*} 
\setminus Y^{j^*}_c$, player~1 plays her winning strategy for the
disjunctive co-Büchi objective on $\game_c^{j^*}$.
(2) If a vertex of $Y^{j^*}_c = \ato(\game^{j^*}, U_c^{j^*})$ is reached, player~1
follows the corresponding attractor strategy until $U_c^{j^*}$ is reached. Then 
player~1 increases the counter by one, or sets the counter to~1 if its value was
$k_2$, and continues with (1).
As the play stays within $V_1^{j^*}$, one of two cases must happen:
Case~1: After some prefix of the play for some counter value $c$ the set 
$\ato(\game^{j^*}, U_c^{j^*})$ is never reached. Then the play satisfies the 
disjunctive co-Büchi objective $\bigvee_{t=1}^{k_1} \objsty{coBüchi}{L_t}$ and 
thus the GR(1) objective. Case~2: For all $c \in \set{1, \ldots k_2}$
the set $U_c^{j^*}$ is reached infinitely often. Then the play satisfies the
generalized Büchi objective $\bigwedge_{\ell=1}^{k_2} \objsty{Büchi}{U_\ell}$ and 
thus the GR(1) objective.
\end{proof}

We show next that whenever Algorithm~\ref{alg:GR1GameBasic} removes vertices from the game graph,
these vertices are indeed winning for player~2. This is due to Lemma~\ref{lem:dominionGR1ConjBuchi} that states that these sets are 
dominions in the current game graph, and Lemma~\ref{lem:attr}(\ref{sublem:subgraph})
that states that all player-2 dominions of the current game graph $\game^j$
are also winning for player~2 in the original game graph $\game$.

\begin{proposition}[Completeness Algorithm~\ref{alg:GR1GameBasic}]\label{prop:p2win_GR1Basic}
Let $V^{j^*}$ be the set of vertices returned by Algorithm~\ref{alg:GR1GameBasic}.
Each vertex in $V \setminus V^{j^*}$ is winning for player~2.
\end{proposition}
\begin{proof}
By Lemma~\ref{lem:attr}(\ref{sublem:subgraph}) it is sufficient to show that
in each iteration~$j$ with $S^j \ne \emptyset$ player~2 has a winning strategy
from the vertices in $S^j$ in $\game^j$. 
Let $j$ be such that $S^j = W_1\left(\overline{\game^j \setminus Y^j_\ell}, \bigwedge_{t=1}^{k_1} \objsty{Büchi}{L^j_t \setminus Y^j_\ell} \right)$.
We first show that $S^j$ is also a player-1 dominion for the generalized Büchi
game on the game graph $\overline{\game^j}$ that includes the vertices of  $Y^j_\ell$,
i.e., that $S$ is a player-2 dominion on $\game^j$. 
By Lemma~\ref{lem:attr}(\ref{sublem:complattr}) the set 
$V^j \setminus Y^j_\ell$ is 1-closed in $\game^j$, i.e., it is  2-closed in $\overline{\game^j}$. 
Thus each 1-dominion of $\overline{\game^j \setminus Y^j_\ell}$ is also 
2-closed in $\overline{\game^j}$ and hence a 1-dominion in $\overline{\game^j}$ 
(see also Lemma~\ref{lem:attr}(\ref{sublem:winclosed})). 
Now as $S^j$ does not contain 
any vertices of $U_\ell$, it is{ by Lemma~\ref{lem:dominionGR1ConjBuchi}} a player-2 dominion in $\game$ w.r.t.\ the GR(1) objective.
Finally, from the above and Lemma~\ref{lem:attr}(\ref{sublem:attrstr}) we have that also $\att(\game^j, S^j)$
is a player-2 dominion in $\game$ with the GR(1) objective.
\end{proof}

Finally the runtime of Algorithm~\ref{alg:GR1GameBasic} is the product of the number of iterations of the nested loops and the runtime
for the generalized Büchi algorithm.

\begin{proposition}[Runtime Algorithm~\ref{alg:GR1GameBasic}]
Algorithm~\ref{alg:GR1GameBasic} runs in $O(k_2 \cdot n \cdot B)$ where $B$ is the runtime bound for the used ConjBüchi algorithm,
i.e., if we use Algorithm~\ref{alg:ConjBuchi} to solve ConjBüchi, the bound is $O(k_1 \cdot k_2 \cdot n^3)$.
\end{proposition}
\begin{proof}
  As in each iteration of the outer loop, except the last one, at least one
  vertex is removed from the maintained graph, there are only $O(n)$ iterations.
  In the inner loop we have $k_2$ iterations, each with a call to the 
  generalized Büchi game algorithm. 
  Thus, in total we have a running time of $O(k_2 \cdot n \cdot B)$.
\end{proof}
} 

{
   % ProgressMeasureAlg
\subsection{Progress Measure Algorithm for Finding Small Dominions}
{
Our goal for the remaining part Section~\ref{sec:GR1Alg}
is to speed up the basic algorithm
by computing ``small'' player-2 dominions faster such that in each iteration
 of the algorithm a ``large'' 2-dominion is found and thereby the number of 
iterations of the algorithm is reduced.
To compute small dominions we use a \emph{progress measure} for generalized Büchi 
games
which is a special instance of the more general progress measure for GR(1) games presented 
in~\cite[Section 3.1]{JuvekarP06}, which itself is based on \cite{Jurdzinski00}. 
In this section w}e first restate the progress measure of \cite{JuvekarP06} in our notation and simplified to generalized
Büchi, then adapt it to not compute the winning sets but dominions of a given size, and finally give an efficient algorithm
to compute the progress measure.

The progress measure of \cite{JuvekarP06} is defined as follows.
Let $\bigwedge_{\ell=1}^{k} \mbox{Büchi}(\target_i)$ be a generalized Büchi objective.
For each $1 \leq \ell \leq k$ we define the value $\bar{n}_\ell$ to be 
$\bar{n}_\ell=|V \setminus \target_\ell|$ 
and a function $\rho_\ell: V \rightarrow \{0,1,\dots,\bar{n}_\ell,\infty\}$. 
The intuitive meaning of a value $\rho_\ell(v)$ is the number of moves player~1 
needs, when starting in $v$, 
to reach a vertex of $\target_\ell \cap W_1$, i.e., $\rho_\ell(v)$ will 
equal to the rank $rank_\pl(\game,\target_\ell \cap W_1,v)$.
As there are only $\bar{n}_\ell$ many vertices which are not in $\target_\ell$,
one can either reach them within $\bar{n}_\ell$ steps or cannot reach them at all.

The actual value $\rho_\ell(v)$ is defined in a recursive fashion via the values of the successor vertices of $v$.
That is, for $v \not\in \target_\ell$ we define $\rho_\ell(v)$ by the values $\rho_\ell(w)$ for $(v,w) \in E$.
Otherwise, if $v \in \target_\ell$, then we already reached $\target_\ell$ and we only have to check whether $v$ is in the winning set.
That is, whether $v$ can reach a vertex of the next target set $\target_{\ell \oplus 1}$ that is also in the winning set $W_1$.
Hence, for $v \in \target_\ell$ we define $\rho_\ell(v)$ by the values $\rho_{\ell \oplus 1}(w)$ for $(v,w) \in E$, where 
$\ell \oplus 1 = \ell + 1$ if $\ell <k$ and $k \oplus 1=1$ and analogously 
$\ell \ominus 1 = \ell - 1$ if $\ell > 1$ and $1 \ominus 1 = k$.
For $v \in V$ one considers all the successor vertices and their values and then picks the minimum if $v \in V_1$
or the maximum if $v \in V_2$. 
In both cases $\rho_\ell(v)$ is set to this  value increased by $1$ if $v \not\in \target_\ell$.
If $v \in \target_\ell$, the value is set to $\infty$ if the minimum (resp. maximum) over the successors is $\infty$
and to $0$ otherwise.
This procedure is formalized via two functions. 
First, $\best(v)$ returns the value of the best neighbor for the player owning~$v$.
\[
 \best(v)= \begin{cases} 
            \min_{(v,w) \in E} \rho_{\ell \oplus 1}(w) & \cif v \in V_1 \wedge v \in \target_\ell\,,\\
            \min_{(v,w) \in E} \rho_{\ell}(w) & \cif v \in V_1 \wedge v \not\in \target_\ell\,,\\
            \max_{(v,w) \in E} \rho_{\ell \oplus 1}(w) & \cif v \in V_2 \wedge v \in \target_\ell\,,\\
            \max_{(v,w) \in E} \rho_{\ell}(w) & \cif v \in V_2 \wedge v \not\in \target_\ell\,.\\
           \end{cases}
\]

Second, the function $\incr$ formalizes the incremental step described above.
To this end, we define for each set $\{0,1,\dots,\bar{n}_\ell,\infty\}$ the unary $\pp$ operator as $x\pp=x+1$ for $x<\bar{n}_\ell$ and $x\pp=\infty$ otherwise.
\[
 \incr(x)= \begin{cases} 
            0  & \cif v \in \target_\ell \wedge x \not= \infty\,,\\
            x\pp & otherwise.
           \end{cases}
\]

The functions $\rho_\ell(.)$ are now defined as the least fixed-point of the operation that updates
all $\rho_\ell(v)$ to $\max(\rho_\ell(v),\incr(\best(v)))$.
The least fixed-point can be computed via the lifting algorithm~\cite{Jurdzinski00}, that starts with all the $\rho_\ell(.)$ initialized as the zero functions and iteratively updates $\rho_\ell(v)$ to $\incr(\best(v))$, for all $v \in V$, until the least fixed-point is reached. 

Given the progress measure, we can decide the generalized Büchi game by 
the following theorem. Intuitively, player~1 can win starting from a vertex with 
$\rho_1(v)<\infty$ by keeping a counter~$\ell$ that is initialized to $1$,
choosing the outgoing edge to $\best(v)$ whenever at a vertex of $V_1$, and 
increasing the counter with $\oplus 1$ when a vertex of $\target_\ell$ is reached.
\begin{theorem}{\cite[Thm.~1]{JuvekarP06}}
 Player~1 as a winning strategy from a vertex $v$ iff $\rho_1(v)<\infty$.
\end{theorem}

As our goal is to compute small dominions, say of size $h$, instead of the whole winning set, 
we have to modify the above progress measure as follows.
In the definition of the functions $\rho_\ell$ we
redefine the value $\bar{n}_\ell$ to be  $\min \{h-1,|V \setminus \target_\ell|\}$
instead of $|V \setminus \target_\ell|$.
The intuition behind this is that if the dominion contains at most $h$~vertices,
then from each vertex in the dominion we can reach each set $\target_\ell$
within $h-1$ steps and we do not care about vertices with a larger distance.

With Algorithm~\ref{alg:ConjBuchiLifting} we give an 
$O(k\cdot h \cdot m)$-time realization of the lifting algorithm for
computing the functions $\rho_\ell$.
It is a 
corrected 
version of the lifting algorithm in~\cite[Section 3.1]{JuvekarP06},
tailored to generalized Büchi objectives and dominion computation,
and exploits ideas from the lifting algorithm in \cite{EtessamiWS05}.
We iteratively increase the values $\rho_\ell(v)$ for all pairs $(v,\ell)$.
The main idea for the runtime bound is to consider each pair  $(v,\ell)$ at most $h$ times and each time we consider a pair we increase the value of $\rho_\ell(v)$ and 
only do computations in the order of the degree of $v$.
To this end, we maintain a list of pairs $(v,\ell)$ for which $\rho_\ell(v)$ must be increased because of some update on $v$'s 
neighbors. We additionally maintain $B_\ell(v)$, which stores the value of $\best(v)$ from the last time we updated $\rho_\ell(v)$, and 
a counter $C_\ell(v)$ for $v \in V_1$, which stores the number of successors $w \in \Out(v)$ with $\rho_\ell(w)=B_\ell(v)$.
Moreover, in order to initialize $C_\ell(v)$ when $B_\ell(v)$ is updated, we use the function $\cnt(v)$ counting the number of successor vertices that have minimal $\rho_\ell$. Notice that for $v \in \target_\ell$ we only distinguish whether $\rho_{\ell \oplus 1}(v)$ is finite or not.
\[
 \cnt(v)=\begin{cases}          
          \left|\{w \in \Out(v) \mid \rho_{\ell \oplus 1}(w)< \infty \}\right| & \cif v \in \target_\ell\,,\\
          \left|\{w \in \Out(v) \mid \rho_\ell(w)=\best(v)\}\right| & \cif v \not\in \target_\ell\,.
         \end{cases}
\]
\begin{algorithm}[t]
	\SetAlgoRefName{GenBuchiProgressMeasure}
	\caption{Lifting Algorithm for Generalized Büchi Games} 
	\label{alg:ConjBuchiLifting}
	\SetKwInOut{Input}{Input}
	\SetKwInOut{Output}{Output}
	\SetKwData{old}{old}
	\BlankLine
	\Input{Game graph $\game=((V, E),(\vo,\vt))$, objective $\bigwedge_{1 \le \ell \le k} \objsty{Büchi}{\target_\ell}$, integer $h \in [1, n]$}
	\Output
	{
	  player~1 dominion / winning set for player~1 if $h=n$
	}
	\BlankLine
	\ForEach{$v \in V$, $1\leq \ell \leq k$}
	    {
	      $B_\ell(v) \gets 0$\;
	      \lIf{$v \in V_1$}{$C_\ell(v) \gets \OutDeg(v)$}
	      $\rho_\ell(v) \gets 0$\;
	    }
	
	$L \gets \{(v,\ell) \mid v \in V, 1 \leq \ell \leq k, v \notin \target_\ell\}$\;
	\While
	  {
	    $L \not= \emptyset$
	  }
	  {
	    pick some $(v, \ell) \in L$ and remove it from $L$\;
	    $\old \gets \rho_\ell(v)$\;
	    $B_\ell(v) \gets \best(v)$\;
	    \lIf{$v \in V_1$}{$C_\ell(v) \gets \cnt(v)$}
	    $\rho_\ell(v) \gets \incr(\best(v))$\;
	    \ForEach{$w \in In(v) \setminus \target_\ell$ with $(w,\ell)\not\in L$, $\rho_\ell(w)<\infty$}
	      {
		{
		  \If{$w \in V_1, \old=B_\ell(w)$}
		    {
			$C_\ell(w)\gets C_\ell(w)-1$\label{lpm:countminus}\;
		        \lIf{$C_\ell(w)=0$}{$L\gets L \cup \{(w,\ell)\}$\label{lpm:addzero}}
		    }
		  \lElseIf{$w \in V_2, \rho_\ell(v)>B_\ell(w)$}{$L\gets L \cup \{(w,\ell)\}$\label{lpm:addv2}}
		}
	      }
	    \If{$\rho_\ell(v)=\infty$}
	      {
	        \ForEach{$w \in In(v) \cap \target_{\ell\ominus 1}$ with $(w,\ell\ominus 1)\not\in L$, $\rho_{\ell\ominus 1}(w)<\infty$}
		{
		  \If{$w \in V_1$}
		    {
			$C_{\ell\ominus 1}(w)\gets C_{\ell\ominus 1}(w)-1$\label{lpm:countminus_inf}\;
		        \lIf{$C_{\ell\ominus 1}(w)=0$}{$L\gets L \cup \{(w,\ell\ominus 1)\}$\label{lpm:addzero_inf}}
		    }
		  \lElseIf{$w \in V_2$}{$L\gets L \cup \{(w,\ell\ominus 1)\}$\label{lpm:addv2_inf}}
		}
	      }
	  }
	  \Return $\{v\in V \mid  \rho_\ell(v)<\infty \text{ for some } \ell\}$
\end{algorithm}

Whenever the algorithm considers a pair $(v,\ell)$, it first computes $\best(v)$,
$\cnt(v)$ in $O(\OutDeg(v))$ time, stores these values in $B_\ell(v)$ and $C_\ell(v)$,
and updates $\rho_\ell(v)$ to $\incr(\best(v))$.
It then identifies the pairs $(w,\ell)$, $(w,\ell\ominus 1)$ that are affected by the change of the value $\rho_\ell(v)$ 
and adds them to the set $L$ in $O(\InDeg(v))$ time.

{
In the remainder of the section we prove the following theorem.
\begin{theorem}\label{thm:ConjBuchiLifting}\label{thm:ConjBuchiLiftingCorrectness}
 For a game graph $\game$ and objective $\obj=\bigwedge_{1 \le \ell \le k} \objsty{Büchi}{\target_\ell}$, 
  Algorithm~\ref{alg:ConjBuchiLifting}  is an $O(k\cdot h \cdot m)$ time procedure  
 that either returns a player-1 dominion or the empty set, and,
 if there is at least one player-1 dominion of size $\leq h$
 then returns a player-1 dominion containing all player-1 dominions of size $\leq h$.
\end{theorem}
}
\begin{remark}
 While for the progress measure in~\cite{JuvekarP06} $\rho_\ell(v)<\infty$ for some $1 \le \ell \le k$ is equivalent to $\rho_{\ell'}(v)<\infty$
 for all $1 \leq \ell' \leq k$, this does not hold in general for our modified progress measure $\rho$.
 Thus we consider the set $\{v\in V \mid  \rho_\ell(v)<\infty \text{ for some } \ell\}$ as a player-1 dominion and 
 not just the set $\{v\in V \mid  \rho_1(v)<\infty\}$ .
\end{remark}

The correctness of Algorithm~\ref{alg:ConjBuchiLifting} is by the following invariants that 
are maintained during the whole algorithm.
These invariants show that 
(a) the data structures $L$, $B_\ell$, and $C_\ell$ are maintained correctly, and 
(b) the values $\rho_\ell(v)$ are bounded from  above by 
(i) $\incr(\best(v))$ and 
(ii) by the rank $rank_1(\game,\target_\ell \cap D,v)$ if $v$ is in a dominion $D$ of size $\leq h$.

\begin{invariant}\label{inv}
  The while loop in Algorithm~\ref{alg:ConjBuchiLifting} has the following loop invariants.
  \begin{enumerate}[\lu1\ru]
      \item For all $v \in V$ and all $1 \leq \ell \leq k$ we have $\rho_\ell(v) \leq \incr(\best(v))$. \label{inv:Inv1}
      \item For all $v \in V$ and all $1 \leq \ell \leq k$ we have that
      if $\rho_\ell(v) \not=0$ or $v \in \target_\ell$,
      then $\rho_\ell(v) = \incr(B_\ell(v))$.\label{inv:Inv2}
      \item For $v \in V_1$ we have $C_\ell(v)=
		\begin{cases}          
		      \left|\{w \in \Out(v) \mid \rho_{\ell \oplus 1}(w)< \infty \}\right| & \cif v \in \target_\ell\,,\\
		      \left|\{w \in \Out(v) \mid \rho_\ell(w)=B_\ell(v)\}\right| & \cif v \not\in \target_\ell, \rho_\ell(v) < \infty\,.
		\end{cases}$\label{inv:Inv3}      
      \item The set $L$ consists exactly of the pairs $(v,\ell)$ with
      $\rho_\ell(v) < \incr(\best(v))$.\label{inv:Inv4}
      \item For each player-1 dominion $D$ with $|D| \leq h$, for each $v \in D$ and all $1 \leq \ell \leq k$ we have 
      $\rho_\ell(v) \leq rank_1(\game,\target_\ell \cap D,v) < h$.\label{inv:Inv5}
  \end{enumerate}
\end{invariant}
\noindent 
Notice that when the algorithm terminates we have
by the Invariants (\ref{inv:Inv1}) \& (\ref{inv:Inv4}) that
$\rho_\ell(v) = \incr(\best(v))$ for all $v \in V$ and all $1 \leq \ell \leq k$, i.e.,
the functions $\rho_\ell(v)$ are a fixed-point for the $\incr(\best(v))$ updates.
By the following lemmata we prove that the above loop invariants are valid.

\begin{lemma}
 After each iteration of the  while loop in Algorithm~\ref{alg:ConjBuchiLifting}  
 we have $\rho_\ell(v) \leq \incr(\best(v))$, for all $v \in V$ and all $1 \leq \ell \leq k$.
\end{lemma}
\begin{proof}
  As all $\rho_\ell(v)$ are initialized to $0$ and $0$ is the minimum value, the inequalities are all satisfied
  in the \emph{base case} when the algorithm first enters the the while loop.
  Now for the \emph{induction step} consider an iteration of the loop and assume the invariant is satisfied before the loop.
  The value $\rho_\ell(v)$ is only changed when the pair $(v,\ell)$ is processed and
  then it is set to $\incr(\best(v))$. Thus the invariant is satisfied after these iterations.
  In all the other iterations with different pairs $(v',\ell')$ the values $\rho_{\ell'}(v')$ are either unchanged or increased. 
  As $\incr(\best(v))$ is monotonic in the values of the neighbors, this can only increase the right side
  of the inequality and thus this invariant is also satisfied after these iterations.
  Hence, if the invariant is valid before an iteration of the loop, it is also valid afterwards.
\end{proof}

\begin{lemma}
 After each iteration of the  while loop in Algorithm~\ref{alg:ConjBuchiLifting}  
 we have that 
 if $\rho_\ell(v) \not=0$ or $v \in \target_\ell$, then $\rho_\ell(v) = \incr(B_\ell(v))$, 
 for all $v \in V$ and all $1 \leq \ell \leq k$.      
\end{lemma}
\begin{proof}
As $\rho_\ell(v)$ is initialized to $0$, this is trivially satisfied in the \emph{base case}.
Now for the \emph{induction step} consider an iteration of the loop and let us assume the invariant is satisfied before the loop.
The values $\rho_\ell(v)$, and $B_\ell(v)$ are only changed when the pair $(v,\ell)$ is processed and
then the invariant is trivially satisfied by the assignments in line 9 and line 11 of the algorithm.
\end{proof}

\begin{lemma}
 After each iteration of the  while loop in Algorithm~\ref{alg:ConjBuchiLifting}  
 for $v \in V_1$ we have $$C_\ell(v)=
		\begin{cases}          
		      \left|\{w \in \Out(v) \mid \rho_{\ell \oplus 1}(w)< \infty \}\right| & \cif v \in \target_\ell\,,\\
		      \left|\{w \in \Out(v) \mid \rho_\ell(w)=B_\ell(v)\}\right| & \cif v \not\in \target_\ell, \rho_\ell(v) < \infty\,.
		\end{cases}$$
\end{lemma}
\begin{proof}
  As \emph{base case} consider the point where the algorithm first enters the while loop.
  All $\rho_\ell(v)$ and $B_\ell(v)$ are initialized to $0$ and thus in both cases the right side of the invariant 
  is equal to $\OutDeg(v)$, which is exactly the value assigned to $C_\ell(v)$.
  
  Now for the \emph{induction step} consider an iteration of the loop and let us assume the Invariants~(\ref{inv:Inv1})--(\ref{inv:Inv3}) are satisfied before the loop.
  Let $v \in V_1$. 
  In an iteration where $(v,\ell)$ is processed in line 10 we set $C_\ell(v)$ to $\cnt(v)$ and 
  hence the invariant is satisfied by the definition of $\cnt(v)$. 
  Otherwise the condition for $C_\ell(v)$ is only affected if a vertex $u \in \Out(v)$ is processed.
  We distinguish the two cases where $v \in \target_\ell$ and where $v \not\in \target_\ell$.
  \begin{itemize}
    \item If $v \in \target_\ell$ then $C_\ell(v)$ is only affected in iterations where pairs $(u,\ell \oplus 1)$
	  are considered.
	  If the updated value of $\rho_{\ell \oplus 1}(u)$ is less than $\infty$ then the set
	  $\{w \in \Out(v) \mid \rho_{\ell \oplus 1}(w)< \infty\}$ is unchanged and also $C_\ell(v)$ is not changed by the algorithm,
	  i.e., the invariant is still satisfied.
	  Otherwise if the updated value of $\rho_{\ell \oplus 1}(u)$ is $\infty$ then 
	  $u$ drops out from the set $\{w \in \Out(v) \mid \rho_{\ell \oplus 1}(w)< \infty\}$ but 
	  also the algorithm decreases $C_\ell(v)$ by one,
	  i.e., again the invariant is  satisfied.
    \item If $v \not\in \target_\ell$ and $\rho_\ell(v) < \infty$ then $C_\ell(v)$ is only affected in iterations where pairs $(u,\ell)$
	  are considered.
	  Let $\rho^o_\ell(u)$  be the value of $\rho_\ell(u)$ before its update.
	  If $\rho^o_\ell(u) > B_\ell(v)$ then $u \not\in \{w \in \Out(v) \mid \rho_\ell(w)=B_\ell(v)\}$
	  and thus the set is not affected by the increased value of $\rho_\ell(u)$. 
	  In that case the algorithm does not change $C_\ell(v)$ and thus 
	  the invariant is satisfied.
	  Otherwise, if $\rho^o_\ell(u) = B_\ell(v)$, then $u \in \{w \in \Out(v) \mid \rho_\ell(w)=B_\ell(v)\}$ before the iteration
	  but not after the iteration. 
	  In that case the algorithm decreases $C_\ell(v)$ by one and thus the invariant is still satisfied.
	  
	  Notice that by Invariants~(\ref{inv:Inv1}) and (\ref{inv:Inv2}), it cannot happen 
	  that $\rho^o_\ell(u) < B_\ell(v)$. To see this, assume by contradiction 
	  $\rho^o_\ell(u) < B_\ell(v)$. 
	  Let $\best^o(v)$ denote the value of 
	  $\best(v)$ before the update of $\rho_\ell(u)$. By~(\ref{inv:Inv1})
	  we have $\rho_\ell(v) \le \incr(\best^o(v))$, by the definition of 
	  $\best^o(v)$ and $v \in V_1 \setminus \target_\ell$ we have 
	  $\best^o(v) \le \rho^o_\ell(u)$ and thus by the assumption 
	  $\best^o(v) < B_\ell(v)$. By (\ref{inv:Inv2}) we have 
	  either $\rho_\ell(v) = \incr(B_\ell(v))$ or $\rho_\ell(v) = 0$.
	  In the first case, as $\incr(x)$ is strictly increasing for $x<\infty$, we have
	  $\incr(\best^o(v)) < \incr(B_\ell(v))= \rho_\ell(v)$
	  and thus a contradiction to (\ref{inv:Inv1}).		  
	  In the second case the pair $(v,\ell)$ was not processed yet and we have
	  a contradiction by $B_\ell(v) = 0$.\qedhere
  \end{itemize}
\end{proof}

\begin{lemma}
 After each iteration of the  while loop in Algorithm~\ref{alg:ConjBuchiLifting}  
 we have that 
 the set $L$ consists exactly of the pairs $(v,\ell)$ with $\rho_\ell(v) < \incr(\best(v))$.
\end{lemma}
\begin{proof}
 The set $L$ is initialized in line~5 with all pairs $(v,\ell)$ such that $v \not\in \target_\ell$.
 For all of these vertices we have $\best(v)=0$ and thus $\incr(\best(v))=1$, 
 i.e., $\rho_\ell(v)\! =\! 0 < \incr(\best(v))\!=\!1$.
 Now consider $(v,\ell) \not\in L$, i.e.,  $v \in \target_\ell$. 
 As all $\rho_\ell(v)=0$, we have $\incr(\best(v))=0$ and thus 
 $\rho_\ell(v) = 0 \not< \incr(\best(v)) = 0$.
 Hence, in the \emph{base case} a pair $(v,\ell)$ is in $L$ iff $\rho_\ell(v) = 0 < \incr(\best(v))=1$.
 
  Now for the \emph{induction step} consider an iteration of the loop and let us assume the Invariants~(\ref{inv:Inv1})--(\ref{inv:Inv4}) are
  satisfied before the loop.
  For the pair $(v,\ell)$ processed in the iteration $\rho_\ell(v)$ is set to $\incr(\best(v))$
  and thus it can be removed from $L$.
  Notice that
  (a) the value of $\rho_\ell(w)$ is only changed when a pair $(w, \ell)$  processed and 
  (b) $\incr(\best(w))$ can only increase when other pairs  $(v, \ell)$ are processed.
  Thus we have to show that in an iteration where the algorithm processes the pair $(v, \ell)$
  all pairs $(w, \ell')$ with $\rho_{\ell'}(w) = \incr(\best(w))$  before the iteration 
  and $\rho_{\ell'}(w) < \incr(\best(w))$ after the iteration are added to the set $L$. 
  The only vertices affected by the change of $\rho_\ell(v)$ are those in $\In(v)$ 
  which are either 
  (i) not in $\target_\ell$ or 
  (ii) in $\target_{\ell\ominus 1}$. 
  In the former case only $\rho_\ell$ is affected 
  while in the latter case only $\rho_{\ell\ominus 1}$ is affected.
  Let $\rho^o_\ell(v)$ and $\rho^n_\ell(v)$ be the values before, 
  respectively after, the update on $\rho_\ell(v)$.
  Notice that if $w \not\in \target_\ell$ and $\rho_\ell(w)=0$, then 
  $(w,\ell)\in L$ by the initialization in line~5.
  Thus in the following, by Invariant~(\ref{inv:Inv2}), we can assume that $\rho_\ell(w) = \incr(B_\ell(w))$
  for all $(w,\ell)\not\in L$. We consider the following cases.
  \begin{itemize}
    \item $w \in \In(v) \setminus \target_\ell$ and $w \in V_1$: 
	  Then $\incr(\best(w))> \rho_\ell(w)$ iff all $u \in \Out(w)$ have $\rho_\ell(u) > B_\ell(w)$.
	  As $(w,\ell) \notin L$ we know that before the iteration there  is at least one $u \in\Out(w)$
	  with $\rho_\ell(u) = B_\ell(w)$.
	  In the case $u\not=v$, $B_\ell(w)$ will not be changed during the iteration and thus
	  $\incr(\best(w)) \not> \rho_\ell(w)$.
	  Hence $\incr(\best(w)) > \rho_\ell(w)$ iff $v$ is the only vertex in $\Out(w)$ with 
	  $\rho^o_\ell(v) = B_\ell(w)$.
	  But then, by Invariant~(\ref{inv:Inv3}), $C_\ell(v)=1$ and thus the algorithm will reduce $C_\ell(v)$ to 0 and add
	  $(v,\ell)$ to the set $L$ in lines~\ref{lpm:countminus}--\ref{lpm:addzero}.
    \item $w \in \In(v) \setminus \target_\ell$ and $w \in V_2$: 
	  Then $\incr(\best(w))> \rho_\ell(w)$ iff there is a vertex $u \in \Out(w)$ with $\rho_\ell(u) > B_\ell(w)$.
	  If there would be such an $u \in \Out(w)$ different from $v$ then by the induction hypothesis we already have $(v,\ell) \in L$.
	  Thus we must have that $\rho^n_\ell(v) > B_\ell(w)$
	  and thus $(w, \ell)$ is added to $L$ in line~\ref{lpm:addv2} of the algorithm.
    \item $w \in \In(v) \cap \target_{\ell\ominus 1}$ and $w \in V_1$:  
	  Then $\incr(\best(w))> \rho_{\ell\ominus 1}(w)$ 
	  iff all $u \in \Out(w)$ have $\rho_\ell(u) =\infty$  and $\rho_{\ell\ominus 1}(w)=0$.
	  This is the case iff $v$ was the only vertex in $\Out(w)$ with $\rho_\ell(v)<\infty$.
	  But then, Invariant \ref{inv:Inv3}, $C_\ell(v)=1$ and thus the algorithm will decrement $C_\ell(v)$ to 0 and add
	  $(v,\ell\ominus 1)$ to the set $L$ in lines~\ref{lpm:countminus_inf}--\ref{lpm:addzero_inf}.	
	  
    \item $w \in \In(v) \cap \target_{\ell\ominus 1}$ and $w \in V_2$: Then $\incr(\best(w))> \rho_{\ell\ominus 1}(w)$ 
	  iff there is an $u \in \Out(w)$ with 
	  $\rho_\ell(u) = \infty $ and $\rho_{\ell\ominus 1}(w) = 0$.
	  If there would be such an $u \in \Out(w)$ different from~$v$ then by the induction hypothesis we already have $(v,\ell\ominus 1) \in L$.
	  Thus, we have that $\rho^n_\ell(v) =\infty > \rho_{\ell\ominus 1}(w)$ and 
	  $\incr(\rho^n_\ell(v)) = \infty > \rho_{\ell\ominus 1}(w)=0$.		   
	  In that case $(w, \ell\ominus 1)$ is added to $L$ in line~\ref{lpm:addv2_inf} of the algorithm.
	  \qedhere
  \end{itemize} 
\end{proof}

\begin{lemma}
For each player-1 dominion $D$ with $|D| \leq h$, 
for each $v \in D$, and 
all $1 \leq \ell \leq k$ we have 
      $\rho_\ell(v) \leq rank_1(\game,\target_\ell \cap D,v) < h$.
\end{lemma}

\begin{proof}
  As all functions $\rho_\ell(.)$ are initialized with the $0$-function, the invariant is satisfied trivially in the
  \emph{base case} when the algorithm first enters the while loop.

  Now for the \emph{induction step} consider an iteration of the loop and let us assume all the invariants are satisfied before the loop.
  First, notice that as $|D|\leq h$, 
  we have $rank_1(\game,\target_\ell \cap D,v)< h$  for all $1 \leq \ell \leq k$ and $v \in D$. 
  The value $\rho_\ell(v)$ is only updated in line 11 and there set to $\incr(\best(v))$.
  We distinguish three different cases.
    \begin{itemize}
      \item Assume $v \in V_1$ and $rank_1(\game,\target_\ell \cap D,v)= j$ with $1 \leq j < h$ then,
	    by definition of $rank_1$, there is a $w \in D, w\not=v$, with $(v,w)\in E$ and 
	    $rank_1(\game,\target_\ell \cap D,w)= j-1$.
	    Now as the invariant is valid before the iteration and $\rho_\ell(w)$
	    is not changed during the iteration, we have $\rho_\ell(w) \leq j-1$ and thus $\best(v) \leq j-1$.
	    Hence, $\incr(\best(v)) \leq j$ and the invariant is still satisfied.
      \item Assume $v \in V_2$ and $rank_1(\game,\target_\ell \cap D,v)= j$ with $1 \leq j < h$ then,
	    by definition of $rank_1$, $rank_1(\game,\target_\ell \cap D,w)= j-1$ for each $(v,w)\in E$
	    (as $D$ is 2-closed we have $w \in D$).
	    Now as the invariant is valid before the iteration and $\rho_\ell(w)$
	    is not changed during the iteration, we have $\rho_\ell(w) \leq j-1$ 
	    for each $(v,w)\in E$ and thus $\best(v) \leq j-1$. 
	    Hence, $\incr(\best(v)) \leq j$ and the invariant is still satisfied.
      \item Finally, assume $rank_1(\game,\target_\ell \cap D,v)= 0$, that is $v \in \target_\ell$.
	    By the induction hypothesis for all $w \in D$ with $(v,w)\in E$ it holds that 
	    $\rho_{\ell \oplus 1}(w)<h$ (and there exists such a $w \in D$)
	    and thus $\best(v)<h$.
	    Hence, $\incr(\best(v))=0$ and the loop invariant is still satisfied.
    \end{itemize}
    Hence, this loop invariant is maintained during the whole algorithm.
\end{proof}

So far we have shown that the algorithm behaves as described by Invariant~\ref{inv}.
The next lemma provides the ingredients to show that the set $W=\{v \in V \mid \rho_\ell(v)<\infty \text{ for some } \ell\}$
is a player-1 dominion by exploiting the fact that the functions $\rho_\ell$ form a fixed-point of the update operator.

\begin{lemma}\label{lem:Conj_ProgressMeasureII}
 Let $W=\{v \in V \mid \rho_\ell(v)<\infty \text{ for some } \ell\}$ be the set computed by Algorithm~\ref{alg:ConjBuchiLifting}.
 \begin{enumerate}[\lu1\ru]
  \item For all $v \in V$: If $\rho_\ell(v) < \infty$,  then
	player~1 has a strategy to reach $\{v' \in \target_\ell \mid \rho_\ell(v')=0\}$ 
	from $v$ by only visiting vertices in $W$. 
  \item For all $v \in  \target_\ell$: If $\rho_\ell(v)=0$, then 
	player~1 has a strategy to reach $\{v' \in \target_{\ell \oplus 1} \mid \rho_{\ell \oplus 1}(v')=0\}$ from $v$ by only visiting vertices in $W$.   
 \end{enumerate}
\end{lemma}
\begin{proof}
Notice that by the Invariants (\ref{inv:Inv1}) \& (\ref{inv:Inv4}) we have 
$\rho_\ell(v) = \incr(\best(v))$ for all $v \in V$ and all $1 \leq \ell \leq k$, i.e.,
the functions $\rho_\ell(v)$ are a fixed-point for the $\incr(\best(v))$ updates.

\smallskip\par\noindent (1)
Consider a vertex $v \in V$ with $\rho_\ell(v)=j$ for $0<j<h$. 
We will show by induction in $j$ that then 
player~1 has a strategy to reach $S=\{v' \in \target_\ell \mid \rho_\ell(v')=0\}$ from $v$ by only visiting vertices in $W$. 
For the \emph{base case} we exploit that the functions $\rho_\ell(v)$ are a fixed-point of the $\incr(\best(v))$ updates.
By the definition of $\incr$ we have that $\rho_\ell(v)=0$ only if 
$v \in \target_\ell$ \footnote{Recall that we assume that each vertex
has at least one outgoing edge.}  and thus we already have reached $S$ in the base case.

For the \emph{induction step} let us assume the claim holds for all $j' < j$ and 
consider a vertex~$v$ with $\rho_\ell(v)=j$.
We distinguish the cases $v \in V_1$ and $v \in V_2$.
\begin{itemize}
 \item $v \in V_1$: Since $\rho$ is a fixed-point of $\incr(\best(v))$, we have that there is at least one vertex~$w$ with $(v,w)\in E$
       and $\rho_\ell(w)=j-1$. By the induction hypothesis, player~1
       has a strategy to reach $S$ starting from $w$,
       and, as player~1 can choose the edge $(v,w)$, also a strategy starting from $v$.
 \item $v \in V_2$: Since $\rho$ is a fixed-point of $\incr(\best(v))$, we have that $\rho_\ell(w) < j$ for all vertices~$w$ with $(v,w)\in E$.
       By the induction hypothesis player~1 has a strategy to reach $S$ starting from any $w$ with $(v,w)\in E$,
       and thus also when starting from $v$.
\end{itemize}
Moreover, in both cases only the vertex $v$ is added to the path induced by the strategy, which by definition is in~$W$.
Hence, in both cases player~1 has a strategy to reach $S$ from $v$ by only visiting vertices in~$W$,
which concludes the proof of part~1.

\smallskip\par\noindent (2) Recall that we have $v \in T_\ell$ and $\rho_\ell(v) = 0$.
Let $S' = \{v' \in \target_{\ell \oplus 1} \mid \rho_{\ell \oplus 1}(v')=0\}$. Again we distinguish whether $v \in V_1$ or $v \in V_2$.
\begin{itemize}
 \item If $v \in V_1$, then, as the functions $\rho_\ell$ form a fixed-point, there is at least one vertex~$w$ with $(v,w)\in E$
       and $\rho_{\ell \oplus 1}(w)< \infty$. 
       Then by (1) player~1 has a strategy to reach $S'$ starting from $w$,
       and, as player~1 can choose the edge $(v,w)$, also a strategy starting from~$v$.
 \item If $v \in V_2$, then, as $\rho$ is a fixed-point, we have $\rho_{\ell \oplus 1}(w)< \infty$ for all $w$ with $(v,w)\in E$.
       Then by~(1) player~1 has a strategy to reach $S'$ starting from any $w$ with $(v,w)\in E$,
       and thus also when starting from $v$.
\end{itemize}
Again, in both cases only the vertex $v$ is added to the path induced by the strategy, which by definition is in $W$, and thus
in both cases player~1 has a strategy to reach $S'$, which concludes the proof of part~2.
\end{proof}

We are now prepared to prove the correctness of Algorithm~\ref{alg:ConjBuchiLifting}.

\begin{proposition} \label{prop:ConjBuchiLiftingCorrectness}
 For the game graph $\game$ and objective 
 $\bigwedge_{1 \le \ell \le k} \objsty{Büchi}{\target_\ell}$,
 Algorithm~\ref{alg:ConjBuchiLifting} 
 either returns a player-1 dominion or the empty set, and,
 if there is at least one player-1 dominion of size $\leq h$
 then it returns a player-1 dominion containing all player-1 dominions of size $\leq h$.
\end{proposition}
\begin{proof}
 We will show that 
 (1) $W=\{v \in V \mid \rho_\ell(v)<\infty \text{ for some } \ell\}$ is a 
 player-1 dominion and that 
 (2) each player-1 dominion of size $\leq h$ is contained in~$W$.

 \smallskip\par\noindent (1) The following strategy is winning for player~1 and does not leave $W$.
    First, for vertices  $v \in W \setminus \bigcup_{\ell=1}^{k} \target_\ell$ pick some $\ell$ s.t. $\rho_\ell(v)< \infty$ and
    play the strategy given by Lemma~\ref{lem:Conj_ProgressMeasureII}(1) to reach $U_\ell \cap W$.
    The first time a set $U_\ell$ is reached, start playing the strategies given by Lemma~\ref{lem:Conj_ProgressMeasureII}(2)
    to first reach the set $U_{\ell \oplus 1} \cap W$, then the set $U_{\ell \oplus 2} \cap W$ and so on.
    This strategy visits all Büchi sets infinitely often and will never leave the set~$W$.
    That is, $W$ is a player-1 dominion.
 
 \smallskip\par\noindent (2) Consider a player-1 dominion $D$ with $|D|\leq h$. Then,  we have that $rank_1(\game,\target_\ell \cap D,v) \leq h-1$ for all $\target_\ell$
 and all $v \in D$
    and by Invariant (\ref{inv:Inv5}) that $\rho_\ell(v) \leq h-1$ for all $v\in D$.
    That is, each $d\in D$ has $\rho_1(v)<\infty$ and thus $D \subseteq W$.
\end{proof}

Finally, let us consider the runtime of Algorithm~\ref{alg:ConjBuchiLifting}.

\begin{proposition}\label{prop:ConjBuchiLiftingTime}
 Algorithm~\ref{alg:ConjBuchiLifting} runs in time~$O(k\cdot h \cdot m)$.
\end{proposition}
\begin{proof}
 Notice that the functions $\best(v)$ and $\cnt(v)$ can be computed in time $O(\OutDeg(v))$ while $\incr(.)$ is in constant time.
 An iteration of the initial foreach loop takes time $O(\OutDeg(v))$ and, as each $v \in V$ is considered $k$ times,
 the entire foreach loop takes time  $O(k\cdot m)$. 
 The running time of Algorithm~\ref{alg:ConjBuchiLifting} is dominated by the while loop.
 Processing a pair $(v,\ell) \in L$ takes time $O(\OutDeg(v)+\InDeg(v))$. 
 Moreover, whenever $(v,\ell)$ is processed, the value of $\rho_\ell(v)$ is increased by $1$ if $v \not\in \target_\ell$
 or by $\infty$ if $v \in \target_\ell$ and thus each pair can be considered at most $h$ times.
 Hence, for the entire while loop we have a running time of 
 $
  O\left(
      h \cdot \sum_{\ell=1}^{k}\sum_{v \in V} (\OutDeg(v)+\InDeg(v))
  \right)
 $
 which can be simplified to $O(k\cdot h \cdot m)$. 
\end{proof}

}

{\subsection{Our Improved Algorithm for GR(1) Games}}
{In this section we present our $O(k_1 k_2 \cdot n^{2.5})$-time algorithm
for GR(1) games, see Algorithm~\ref{alg:GR1Game}.}
The overall structure of {the algorithm} is the 
same as for the basic algorithm: We search for a player-2 dominion~$S^j$
and if one is found, then its player-2 attractor $D^j$ is determined and removed 
from the current game graph $\game^j$ (with $\game^1 = \game$) to create the 
game graph for the next iteration, $\game^{j+1}$. If no player-2 dominion
exists, then the remaining vertices are returned as the winning set of player~1.
The difference to the basic algorithm lies in the way we search for player-2 dominions.
Two different procedures are used for this purpose: First we search for ``small''
dominions with the subroutine~$\domalg$. If no small dominion exists, 
then we search for player-2 dominions as in the basic algorithm. 
{The guarantee that we find a ``large'' dominion in the second case (if a player-2 dominion exists)
allows us to bound the number of times this can happen.}
{The subroutine $\domalg$ called with parameter $\kmax$ on a game graph~$\game$
provides the guarantee to identify all player-2 dominions~$D$ for which 
$\lvert \att(\game, D) \rvert \le \kmax$, where $\kmax$ is 
set to $\sqrt{n}$ to achieve the desired runtime.

\smallskip\noindent{\em Search for Large Dominions.}
If the subroutine $\domalg$ returns an empty set, i.e., when we have for all player-2
dominions~$D$ that $\lvert \att(\game^j, D) \rvert > \kmax$, then we search 
for player-2 dominions as in the basic algorithm: 
For each $1 \le \ell \le k_2$ first the player-1 attractor~$Y^j_\ell$
of $U^j_\ell$ is temporarily removed from the graph. 
Then a generalized Büchi game with target sets $L^j_1  \setminus Y^j_\ell, \ldots,
L^j_{k_1} \setminus Y^j_\ell$ is solved on $\overline{\game^j \setminus Y^j_\ell}$. 
The generalized Büchi player in this game corresponds
to player~2 in the GR(1) game and his winning set to a player-2 dominion
in the GR(1) game{, see Lemma~\ref{lem:dominionGR1ConjBuchi}}.
}
{
\begin{algorithm}[t!]
		\SetAlgoRefName{GR(1)Game}
	\caption{GR(1) Games in $O(k_1 \cdot k_2 \cdot n^{2.5})$ Time} 
	\label{alg:GR1Game}
	\SetKwInOut{Input}{Input}
	\SetKwInOut{Output}{Output}
	\SetKw{break}{break}
	\BlankLine
	\Input{Game graph $\game=((V, E),(\vo,\vt))$, Obj.\ $\bigwedge_{t=1}^{k_1} \objsty{Büchi}{L_t} \rightarrow \bigwedge_{\ell=1}^{k_2} \objsty{Büchi}{U_\ell}$}
	\Output
	{
	  Winning set of player~1
	}
	\BlankLine
	$\game^1 
	\gets \game$\;
	$\set{U^1_\ell} \gets \set{U_\ell}$; $\set{L^1_t} \gets \set{L_t}$\;
	$j \gets 0$\;
	\Repeat
	  {
	    $D^j= \emptyset$
	  }
	  {
	    $j \gets j+1$ \;
				$S^j \gets \domalg(\game^j, \set{U^j_\ell}, \set{L^j_t}, \sqrt{n})$\;
				\If{$S^j = \emptyset$}{
					    \For{$1 \leq \ell \leq k_2$}
							{
							$Y^j_\ell \gets \ato(\game^j, U^j_\ell)$\;
							$S^j \gets \buchialg(\overline{\game^j \setminus Y^j_\ell}, 
							\bigwedge_{\ell=1}^{k_1} \objsty{Büchi}{L^j_t  \setminus Y^j_{\ell}})$\;
							\lIf{$S^j \not= \emptyset$}{\break}
							}
				}
	    $D^j \gets \att(\game^j, S^j)$\;
	    $\game^{j+1}\gets \game^j \setminus D^j$\;
	    $\set{U^{j+1}_\ell} \gets \set{U^j_\ell \setminus D^j}$; 
	    $\set{L^{j+1}_t} \gets \set{L^j_t \setminus D^j}$\;
	  }
	  \Return $V^j$
	  \BlankLine
		\SetKwProg{myproc}{Procedure}{}{}
		\myproc{$\domalg(\game^j, 
		\set{U^j_\ell}, \set{L^j_t}, \kmax)$}{
			\For{$i \gets 1$ \KwTo $\lceil\log_2(2\kmax)\rceil$}{
				construct $G^j_{i}$\; 
				$Z^j_i \gets \{v \in \vt \mid \OutDeg(G^j_i, v)=0\} \cup 
			     \{v \in \vo \mid \OutDeg(G^j, v)>2^i\}$\;
				\For{$1 \leq \ell \leq k_2$}{
					$Y^j_{i, \ell} \gets \ato(\game^j_{i}, U^j_\ell \cup Z^j_{i})$\;
					$X_{i,\ell}^j \gets \progress(\overline{\game^j_{i} \setminus Y^j_{i, \ell}}, 
					\bigwedge_{\ell=1}^{k_1} \objsty{Büchi}{L^j_t  \setminus Y^j_{i, \ell}}, 2^i)$\;
					\lIf{$X_{i,\ell}^j \ne \emptyset$}{
						\Return{$X_{i,\ell}^j$}
					}
				}
			}
			\Return{$\emptyset$}\;
		}
\end{algorithm}
} 

\smallskip\noindent{\em Procedure $\domalg$.}
The procedure $\domalg$ searches for player-2 dominions 
in the GR(1) game, and returns some dominion if 
there exists a dominion~$D$ with $\lvert \att(\game, D) \rvert \le \kmax$.
To this end we again consider generalized Büchi games with target sets 
$L^j_1, \ldots, L^j_{k_1}$, where the generalized Büchi player corresponds to 
player~2 in the GR(1) game. We use the same hierarchical graph decomposition as for 
Algorithm~\ref{alg:ConjBuchi}: Let the incoming edges of each vertex 
be ordered such that the edges from vertices of $\vt$ come first;
for a given game graph $\game^j$ the graph 
$G^j_i$ contains all vertices of $\game^j$, for each vertex its
first $2^i$ incoming edges, and for each vertex with outdegree at most $2^i$
all its outgoing edges. The set $Z^j_i$ contains all vertices of $\vo$ with 
outdegree larger than $2^i$ in $G^j$ and all vertices of $\vt$ that have no outgoing 
edge in $G^j_i$. We start with $i = 1$ and increase $i$ by one as long as no 
dominion was found. For a given~$i$ we perform the following operations for 
each $1 \le \ell \le k_2$: First the player~1 attractor $Y_{i, \ell}^j$
of $U^j_\ell \cup Z^j_i$ is determined. Then we 
search for player-1 dominions on
$\overline{\game^j_i \setminus Y^j_{i, \ell}}$ w.r.t.\ the objective 
$\bigwedge_{t=1}^{k_1} \objsty{Büchi}{L_t\setminus Y^j_{i, \ell}}$
with the generalized Büchi progress measure algorithm and parameter $h = 2^i$,
i.e., by Theorem~\ref{thm:ConjBuchiLiftingCorrectness} 
the progress measure algorithm returns all generalized Büchi dominions in 
$\overline{\game^j_i \setminus Y^j_{i, \ell}}$ of size at most~$h$.

The following lemma shows how the properties of the hierarchical graph 
decomposition extend{ from generalized Büchi games} to GR(1) games. The 
first part is crucial for correctness: Every non-empty set found by 
the progress measure algorithm on $\overline{\game^j_i \setminus Y^j_{i, \ell}}$
for some~$i$ and~$\ell$ is indeed a player-2 dominion in the GR(1) game.
The second part is crucial for the runtime argument: Whenever the basic 
algorithm for GR(1) games would identify a player-2 dominion~$D$ with 
$\lvert \att(\game, D) \rvert \le 2^i$, then $D$ is also a generalized
Büchi dominion in $\overline{\game^j_i \setminus Y^j_{i, \ell}}$ for some $\ell$.

\begin{lemma}\label{lem:domdecomp}
	Let the notation be as in Algorithm~\ref{alg:GR1Game}.
	\begin{compactenum}
		\item Every $X_{i,\ell}^j \ne \emptyset$ is a 
		player-2 dominion in the GR(1) game on $\game^j$ 
		with $X_{i,\ell}^j \cap U_\ell^j = \emptyset$.\label{sublem:domsound}
		\item If for player~2 
		 there exists in $\game^j$ a dominion~$D$ w.r.t.\ the generalized
		Büchi objective $\bigwedge_{t=1}^{k_1} \objstytxt{Büchi}{L^j_t}$ such that 
		$D \cap U_\ell^j = \emptyset$ for some $1 \le \ell \le k_2$ and
		$\lvert \att(\game^j, D) \rvert \le 2^i$, then $D$ is a dominion
		w.r.t.\ the generalized Büchi objective $\bigwedge_{t=1}^{k_1} 
		\objstytxt{Büchi}{L^j_t \setminus Y^j_{i, \ell}}$ in 
		$\game^j_{i} \setminus Y^j_{i, \ell}$.\label{sublem:domsize}
	\end{compactenum}
\end{lemma}

{
\pagebreak
\begin{proof}We prove the two points separately.
	\begin{enumerate}
		\item By Theorem~\ref{thm:ConjBuchiLiftingCorrectness} the set $X_{i,\ell}^j$ is 
		a player-2 dominion on $\game^j_i \setminus Y^j_{i, \ell}$ w.r.t.\ 
		the generalized Büchi objective $\bigwedge_{t=1}^{k_1} 
		\objsty{Büchi}{L^j_t \setminus Y^j_{i, \ell}}$ of player~2.
		By Lemma~\ref{lem:attr}(\ref{sublem:complattr}) $V^j \setminus Y^j_{i, \ell}$
		is closed for player~1 on $\game^j_i$. Thus by 
		Lemma~\ref{lem:attr}(\ref{sublem:winclosed}) $X_{i,\ell}^j$ is a player-2 
		dominion w.r.t.\ the generalized Büchi objective also in $\game^j_i$.
		As $X_{i,\ell}^j$ is player~1 closed in $\game^j_i$ and 
		does not intersect with $Z_i^j$, it is player~1 closed
		in $\game^j$ by Lemma~\ref{lem:decomp}(\ref{sublem:sound}).
		Thus by $E^j_i \subseteq E^j$, the set $X_{i,\ell}^j$ is a player-2 dominion 
		w.r.t.\ the generalized Büchi objective also in $\game^j$.
		Since $X_{i,\ell}^j$ does not intersect with $U^j_\ell$, it is also a player-2 
		dominion in the GR(1) game on $\game^j${ (cf.\ Lemma~\ref{lem:dominionGR1ConjBuchi})}.
		\item Since every player-2 dominion is player-1 closed, we have by 
		Lemma~\ref{lem:decomp}(\ref{sublem:size}) that \upbr{i} $\game^j[D] = \game^j_i[D]$,
		\upbr{ii} $D$ does not intersect with $Z^j_i$, and \upbr{iii}
		$D$ is player~1 closed in $\game_i^j$. Thus we have 
		that \upbr{a} $D$ does not intersect with $Y^j_{i, \ell}$
		and \upbr{b} player~2 can play the same winning strategy for the 
		vertices in $D$ on $\game_i^j$ as on $\game^j$.
		\qedhere
	\end{enumerate}
\end{proof}
}
{\par}From this we can draw the following two corollaries: (1)~When we had to go 
up to $i^*$ in the graph decomposition to find a dominion, then its attractor
has size at least $2^{i^*-1}$ and (2)~when $\domalg$ returns an empty set,
then all player-2 dominions in the current game graph have more
than $\kmax = \sqrt{n}$ vertices.
{
\begin{corollary}\label{cor:domsize}
Let $j$ be some iteration of the repeat-until loop in Algorithm~\ref{alg:GR1Game}
and consider the call to 
$\domalg(\game^j, \set{U^j_\ell}, \set{L^j_t}, \kmax)$.
\begin{enumerate}
	\item If for some $i > 1$ we have $X_{i,\ell}^j \ne \emptyset$ but 
	$X_{i-1,\ell}^j = \emptyset$, then $\lvert \att(\game^j, X_{i,\ell}^j) 
	\rvert > 2^{i-1}$.\label{sublem:domsizelevel}
	\item If $\domalg(\game^j, 
		\set{U^j_\ell}, \set{L^j_t}, \kmax)$ returns the empty set, then for 
		every player-2 dominion $D$ in the GR(1) game 
		we have $\lvert \att(\game^j, D) \rvert > \kmax$.\label{sublem:domsizekmax}
\end{enumerate}
\end{corollary}
}
{
\begin{proof} We prove the two points separately.
	\begin{enumerate}
		\item By Lemma~\ref{lem:domdecomp}(\ref{sublem:domsound}) $X_{i,\ell}^j$ is a 
		player-2 dominion in the GR(1) game on $\game^j$ with $X_{i,\ell}^j
		\cap U_\ell^j = \emptyset$
		and thus in particular a dominion w.r.t.\ the generalized Büchi objective 
		$\bigwedge_{t=1}^{k_1} \objsty{Büchi}{L^j_t}$ such that 
		$X_{i,\ell}^j \cap U_\ell^j = \emptyset$. Assume by contradiction 
		$\lvert \att(\game^j, X_{i,\ell}^j) \rvert \le 2^{i-1}$. Then by 
		Lemma~\ref{lem:domdecomp}(\ref{sublem:domsize}) we have 
		$X_{i-1,\ell}^j \ne \emptyset$, a contradiction.
		
		\item Assume there exists a 2-dominion $D$ with 
		$\lvert \att(\game^j, D) \rvert \leq \kmax$.
		{Then by Lemma~\ref{lem:domwithoutU} there is also a 
		2-dominion}		$D' \subseteq D$ that meets the criteria 
		of Lemma~\ref{lem:domdecomp}(\ref{sublem:domsize}). 
		Let~$i'$ be the minimal value such that $\lvert \att(\game^j, D') \rvert \leq 2^{i'}$, certainly $i' \leq \lceil \log_2(\kmax) \rceil$.
		Now, by Lemma~\ref{lem:domdecomp}(\ref{sublem:domsize}), we have that
		$D'$ is a dominion w.r.t.\ the generalized Büchi objective $\bigwedge_{t=1}^{k_1} 
		\objsty{Büchi}{L^j_t \setminus Y^j_{i', \ell}}$ in $\game^j_{i'} \setminus Y^j_{i', \ell}$.
		{By the correctness of Algorithm~\ref{alg:ConjBuchiLifting}, }the 
		set $X_{i',\ell}^j$ is a dominion containing $D'$
		and thus $\domalg(\game^j = ((V^j, E^j),(V_1^j, V_2^j)), \set{U^j_\ell}, \set{L^j_t}, \kmax)$
		returns a non-empty set.
		\qedhere
	\end{enumerate}
\end{proof}
}{

For the final game graph $\game^{j^*}$ we can build a winning strategy for player~1 in the same way as for Algorithm~\ref{alg:GR1GameBasic}.
That is, by combining her winning strategies for the disjunctive objective in the subgraphs $\overline{\game_\ell^{j^*}}$
and the attractor strategies for $\ato(\game^{j^*}, U_\ell)$. 
\begin{lemma}[Soundness of Algorithm~\ref{alg:GR1Game}]
Let $V^{j^*}$ be the set of vertices returned by Algorithm~\ref{alg:GR1Game}.
Each vertex in $V^{j^*}$ is winning for player~1.
\end{lemma}
\begin{proof}
{When the algorithm terminates we have $S^{j^*} = \emptyset$. Thus
the winning strategy of player~1 can be constructed in the same way 
as for the set returned by Algorithm~\ref{alg:GR1GameBasic}.
(cf.\ Proof of Proposition~\ref{prop:p1win_GR1Basic})}
\end{proof}
}
Next we show that whenever Algorithm~\ref{alg:GR1Game} removes vertices from the game graph,
these vertices are indeed winning for player~2. This is due to Lemma~\ref{lem:domdecomp}(\ref{sublem:domsound}), stating that these sets are 
2-dominions in the current game graph and Lemma~\ref{lem:attr}, stating that all player-2 dominions of the current game graph $\game^j$
are also winning for player~2 in the original game graph $\game$.
\begin{lemma}[Completeness Algorithm~\ref{alg:GR1Game}]
Let $V^{j^*}$ be the set of vertices returned by Algorithm~\ref{alg:GR1Game}.
Each vertex in $V \setminus V^{j^*}$ is winning for player~2.
\end{lemma}
\begin{proof}
By Lemma~\ref{lem:attr}(\ref{sublem:subgraph}) it is sufficient to show that
in each iteration~$j$ with $S^j \ne \emptyset$ player~2 has a winning strategy
from the vertices in $S^j$ in $\game^j$. 
If a non-empty set $S^j$ is returned by $\domalg$, then $S^j$ 
is winning for player~2 by Lemma~\ref{lem:domdecomp}(\ref{sublem:domsound}).
For the case where $S^j$ is empty after the call to $\domalg$,
{the set $S^j$ is determined in the same way as 
in the basic algorithm for GR(1) games and thus is winning by the correctness of 
Algorithm~\ref{alg:GR1GameBasic} (cf.\ Proof of Proposition~\ref{prop:p2win_GR1Basic}).}
\end{proof}
{
Finally, as the runtime of the subroutine $\domalg$ scales with the size of the smallest player-2 dominion in $\game^j$ and we have 
only make $O(\sqrt{n})$ many calls to $\buchialg$, we obtain a runtime of $O(k_1 \cdot k_2 \cdot n^{2.5})$.
\begin{theorem}[Runtime Algorithm~\ref{alg:GR1Game}]
	The algorithm can be implemented to terminate in $O(k_1 \cdot k_2 \cdot n^{2.5})$ time.
\end{theorem}}
{
\begin{proof}
		We analyze the \emph{total} runtime over all iterations of the repeat-until loop.
	The analysis uses that whenever a player-2 dominion $D^j$ is identified,
	then the vertices of $D^j$ are removed from the maintained game graph. In particular,
	we have that whenever $\domalg$ returns an empty set, either 
	at least $\kmax = \sqrt{n}$ vertices are removed from the game graph or 
	the algorithm terminates. Thus this case can happen at most $O(n / \kmax) = O(\sqrt{n})$
	times. In this case $\buchialg$ is called $k_2$ times. By Proposition~\ref{prop:time_conjbuchi}
	this takes total time $O( \sqrt{n} \cdot k_2 \cdot k_1 \cdot n^2) = O(k_1 k_2 \cdot n^{2.5})$.
	
	We next bound the total time spent in $\domalg$. 
	To efficiently construct the graphs $G^j_i$ and the vertex sets $Z^j_i$ we
	maintain (sorted) lists of the incoming and the outgoing edges of each vertex.
	These lists can be updated whenever an obsolete entry is encountered in the 
	construction of $G^j_i$; as
	each entry is removed at most once, maintaining this data structures takes
	total time $O(m)$.
	Now consider a fixed iteration~$i$ of the outer for-loop in $\domalg$.
	The graph $G_i^j$ has $O(2^i \cdot n)$ edges and thus, given the above data
	structure for adjacent edges, the graphs $G^j_i$ and the 
	sets $Z^j_i$ can be constructed in $O(2^i \cdot n)$ time. 
	Further the $k_2$ attractor computations in the inner for-loop can be done in time 
	$O(k_2 \cdot 2^i \cdot n)$.
	The runtime of iteration~$i$ is dominated by the $k_2$ calls to $\progress$.
	By Theorem~\ref{thm:ConjBuchiLifting} the calls to $\progress$ in iteration~$i$,
	with parameter $h$ set to $2^i$,
	 take time $O(k_1 k_2 \cdot n \cdot 2^{2i})$.
	Let $i^*$ be the iteration at which $\domalg$ stops after it is 
called in the $j$th iteration of the repeat-until loop. The runtime 
for this call to $\domalg$ from $i = 1$ to $i^*$ forms a geometric series that
is bounded by $O(k_1 k_2 \cdot n \cdot 2^{2i^*})$.
By Corollary~\ref{cor:domsize} 
either (1) a dominion $D$ with $\lvert \att(\game^j, D) \rvert > 2^{i^*-1}$ 
vertices was found by $\domalg$ or (2) all 
dominions in $\game^j$ have more than $\kmax$ vertices. Thus either (2a) a dominion~$D$
with more than $\kmax$ vertices is detected in the subsequent call to $\buchialg$ 
or (2b) there is no dominion in $\game^j$ and $j$ is the 
last iteration of the algorithm. Case (2b) can happen at most 
once and its runtime is bounded by $O(k_1 k_2 \cdot n \cdot 2^{2\log(\kmax)}) = 
O(k_1 k_2 \cdot n^2)$. In the cases (1) and (2a) more than $2^{i^*-1}$ vertices 
are removed from the graph in this iteration, as $\kmax > 2^{i^*-1}$. We charge each such 
vertex $O(k_1 k_2 \cdot n \cdot 2^{i^*}) = O(k_1 k_2 \cdot n \cdot \kmax)$ time. 
Hence the total runtime for these cases is $O(k_1 k_2 \cdot n^2 \cdot \kmax) = 
O(k_1 k_2 \cdot n^{2.5})$.
\end{proof}
}
{
\begin{remark}
	Algorithm~\ref{alg:GR1Game} can be modified to additionally return winning 
	strategies for both players. Procedure~$\progress(\game,\obj,h)$ can be 
	modified to return a winning strategy within the returned dominion.
	Procedure~$\buchialg$ can be modified to return winning strategies
		for both player in the generalized Büchi game.
	Thus for player~2 a winning strategy for the dominion~$D^j$ that is identified 
	in iteration~$j$ of the algorithm can be constructed by combining his winning 
	strategy in the generalized Büchi game in which~$S^j$ was identified with his 
	attractor strategy to the set $S^j$. For player~1 we can obtain a winning 
	strategy in the final iteration of the algorithm by combining 
	for $1\le \ell \le k_2$ her attractor strategies to the sets $U_\ell$ with 
	her winning strategies in the generalized Büchi games for each of the game graphs
	$\overline{\game^j_i \setminus Y^j_{i, \ell}}$
\upbr{as described in the proof of Proposition~\ref{prop:p1win_GR1Basic}}. 
\end{remark}
}

 % conclusion
\section{Conclusion}
In this work we consider
the algorithmic problem of computing the winning 
sets for games on graphs with generalized B\"uchi and GR(1) objectives.
We present improved algorithms for both, and conditional lower bounds
for generalized B\"uchi objectives.

The existing upper bounds and our conditional lower bounds are tight for 
(a)~for dense graphs, and (b)~sparse graphs with constant size target sets.
Two interesting open questions are as follows:
(1)~For sparse graphs with $\theta(n)$ many target sets of size $\theta(n)$ 
the upper bounds are cubic, whereas the conditional lower bound is 
quadratic, and closing the gap is an interesting open question. 
(2) For GR(1) objectives we obtain the conditional lower bounds from 
generalized B\"uchi objectives, which are not tight in this case;
whether better (conditional) lower bounds can be established also remains open.

\subparagraph*{Acknowledgements.}
K. C., M. H., and W. D. are partially supported by the Vienna
Science and Technology Fund (WWTF) through project ICT15-003.
K. C. is partially supported by the Austrian Science Fund (FWF)
NFN Grant No S11407-N23 (RiSE/SHiNE) and an ERC Start grant
(279307: Graph Games). For W. D., M. H., and V. L. the research
leading to these results has received funding from the European
Research Council under the European Union’s Seventh Framework
Programme (FP/2007-2013) / ERC Grant Agreement no. 340506.

\printbibliography[heading=bibintoc]

\end{document}